\documentclass[3p,12pt]{elsarticle}

\usepackage{amsmath}
\usepackage{amssymb}
\usepackage{amsthm}
\usepackage{algorithm}
\usepackage{algorithmicx}
\usepackage{algpseudocode}
\usepackage{amsmath}
\usepackage{amssymb}
\usepackage{amsthm}
\usepackage{booktabs}
\usepackage{graphicx}
\usepackage[caption=false,font=footnotesize]{subfig}
\usepackage{url}

\newcommand{\dcsAlgRef}[1]{Algorithm~\ref{alg:#1}} 
\newcommand{\dcsSecRef}[1]{Section~\ref{sec:#1}} 
\newcommand{\dcsFigRef}[1]{Figure~\ref{fig:#1}} 
\newcommand{\dcsTabRef}[1]{Table~\ref{tbl:#1}} 
\newcommand{\dcsEqRef}[1]{Eq.~\eqref{eq:#1}} 
\newcommand{\dcsHRule}{\rule{\linewidth}{0.08em}}
\newtheorem{proposition}{Proposition}
\newtheorem{definition}{Definition}

\journal{XXX}

\begin{document}

\begin{frontmatter}

\begin{center}
{\Huge Maximizing Profit in Green Cellular Networks through Collaborative Games}
\end{center}

\begin{center}
{\Large\bfseries
Please, cite this paper as:
\begin{tabular}{c}
Cosimo Anglano, Marco Guazzone, Matteo Sereno,\\
\emph{``Maximizing Profit in Green Cellular Networks},\\
\emph{through Collaborative Games,''}\\
Computer Networks, Volume 75, Part A, pp. 260--275, 2014.\\
DOI: \url{10.1016/j.comnet.2014.10.003}\\
{\normalsize Publisher: \url{http://www.sciencedirect.com/science/article/pii/S1389128614003582}}
\end{tabular}
}
\end{center}

\newpage

\title{Maximizing Profit in Green Cellular Networks through Collaborative Games\tnoteref{cit,pub}}
\tnotetext[cite]{Please, cite as: \emph{Cosimo Anglano, Marco Guazzone, Matteo Sereno, ``Maximizing Profit in Green Cellular Networks through Collaborative Games,'' Computer Networks, Volume 75, Part A, Pages 260--275, 2014, DOI:10.1016/j.comnet.2014.10.003.}}
\tnotetext[pub]{Link to publisher: \url{http://www.sciencedirect.com/science/article/pii/S1389128614003582}}

\author[unipmn]{Cosimo Anglano\corref{corauth}}
\ead{cosimo.anglano@di.unipmn.it}

\author[unipmn]{Marco Guazzone}
\ead{marco.guazzone@di.unipmn.it}

\author[unito]{Matteo Sereno}
\ead{matteo.sereno@unito.it}

\address[unipmn]{Department of Science and Technological Innovation, University of Piemonte Orientale, Italy}
\address[unito]{Department of Computer Science, University of Torino, Italy}

\cortext[corauth]{Corresponding author}

\begin{abstract}
In this paper, we deal with the problem of maximizing the profit of
\emph{Network Operators} (NOs) of green cellular networks in situations where
\emph{Quality-of-Service} (QoS) guarantees must be ensured to users, and
\emph{Base Stations} (BSs) can be shared among different operators.

We show that if NOs cooperate among them,
by mutually sharing their users and BSs, then each one of them
can improve its net profit.

By using a game-theoretic framework, we study the problem of forming
\emph{stable} coalitions among NOs.
Furthermore, we propose a mathematical optimization model to allocate users to a
set of BSs, in order to reduce costs and, at the same time, to meet user QoS for
NOs inside the same coalition.
Based on this, we propose an algorithm, based on cooperative game theory,
that enables each operator to decide with whom to cooperate in order to
maximize its profit.

This algorithms adopts a distributed approach in which each NO autonomously
makes its own decisions, and where the best solution arises without
the need to synchronize them or to resort to a trusted third party.

The effectiveness of the proposed algorithm is demonstrated through
a thorough experimental evaluation considering real-world traffic traces,
and a set of realistic scenarios.
The results we obtain indicate that our algorithm allows a population of
NOs to significantly improve their profits thanks to the combination
of energy reduction and satisfaction of QoS requirements.
\end{abstract}


\end{frontmatter}


\section{Introduction} \label{sec:intro}

The increasing consumption of electrical energy is one of the most important
issues characterizing modern society because of its effects on climate changes
and on the depletion of non-renewable sources.
In this scenario, the ICT sector plays a key role, being responsible for
about $10\%$ of the world carbon footprint and electrical energy consumption
\cite{GAP-07,ClimateGroup-2008-Smart2020}.

Reportedly \cite{Vereecken-2011-Power,Oh-2011-Toward},
within the ICT sector, the mobile telecommunication industry (and, in particular, cellular networks)
is one of the major contributors to energy consumption. 
This has stimulated the interest towards a new
research area called \emph{green cellular networks}~\cite{DeDomenico-2014-Enabling},
that aims at reducing the energy consumption of these communication infrastructures.

From the perspective of a \emph{cellular Network Operator} (NO), the reduction
of electrical energy consumption is not only a matter of being ``green'' and
responsible, but also an economically important opportunity. 
As a matter of fact, it has been argued that nearly half of the total
operating expenses of a NO is due to energy costs
\cite{Vereecken-2011-Power,Oh-2011-Toward}.
Furthermore, a significant part of these costs are due to \emph{Base Stations} (BSs) \cite{Auer2011}:
indeed, even in the case of little or no activity, a BS can consume more than
$90\%$ of its peak energy \cite{Oh-2011-Toward,Peng-2011-Traffic}.
Thus, by reducing energy consumption, a NO may sensibly increases its profit. 

Consequently, a lot of research effort has been concentrated lately on the reduction of
the energy consumed by BSs. Techniques like the design of more energy-efficient hardware
equipments, or the use of new energy saving techniques (e.g., \emph{sleep
modes} \cite{Micallef-2010-Cell} and \emph{cell zooming} \cite{Niu-2010-Cell})
to switch off under-utilized BSs during low traffic periods and to transfer the
corresponding load to neighboring cells, have been proposed as possible solutions.

Such techniques, however, must be applied with care so as to maintain
\emph{Quality-of-Service} (QoS) guarantees agreed by a NO with its
customers, whose violations imply monetary losses for that NO.
Specifically, since fewer transmission resources are available at a cell when such 
energy-efficient techniques are used, bottlenecks may form for those users
connected to that cell, who may thus experience QoS levels lower than
guaranteed, and in some cases may be even unable to receive service at all.
Finally, the use of techniques like cell zooming may cause other problems,
such as inter-cell interference and coverage holes \cite{Niu-2010-Cell}.

In this paper, we argue that, if NOs cooperate among them
by mutually sharing their users and BSs, then each one of them
can improve its net profit by either (a) reducing energy
costs by switching off its BSs and offloading its users to switched on
BSs of other NOs, or (b) increasing its earnings by attracting users
from other NOs, or by relying on BSs of other NOs to accept more users than what
could do by working alone.

Obviously, it is unreasonable to expect that each NO is willing to unconditionally
cooperate with the other ones regardless the benefits it receives.
Such a cooperation arises indeed only if
suitable benefits result from it, and if the risks of monetary losses
are kept within acceptable limits.

In this paper, we devise a decision algorithm that 
provides a set of NOs with suitable means to decide whether to cooperate with
other NOs, and if so with whom to cooperate.
Our algorithm is based on game-theoretic techniques,
where the process of establishing cooperation among the NOs is modeled as
a \emph{cooperative game with transferable utility} \cite{Peleg-2007-CooperativeGames}
(in particular, as a \emph{hedonic game} \cite{DrezeGreenberg1980}, whereby each
NO bases its decision on its own preferences).

More specifically, we propose a game-theoretic framework to study the problem of
forming \emph{stable} coalitions among NOs, and a mathematical optimization
model to allocate users to a set of BSs, in order to reduce costs and, at the
same time, to meet user QoS for NOs inside the same coalition.
We achieve our goal by devising
a \emph{hedonic shift algorithm} to form stable coalitions that allows
each NO to autonomously and selfishly decide whether to leave the current
coalition to join a different one or not on the basis of the net profit it
receives for doing so.

In our approach, each NO pays for the energy consumed to serve each user,
whether it belongs to it or to another NO, but receives a payoff (computed as
discussed later) for doing so.
We prove that the proposed algorithm converges to a \emph{Nash-stable} set of
disjoint coalitions \cite{BOGOMONLAIA-JACKSON_2002}, whereby no NO can benefit
to leave the current coalition to join a different one.

Our solution adopts an asynchronous approach in which each NO autonomously 
 makes its own decisions, and where the best solution arises without
the need to synchronize them or to resort to a trusted third party.
As a consequence, the solution we propose can be readily
implemented in a distributed fashion.

To demonstrate the effectiveness of the algorithm we propose, we carry out 
a thorough experimental evaluation considering real-world traffic traces,
and a set of realistic scenarios. 
The results we obtain indicate that our algorithm allows indeed a population of
NOs to significantly improve their profits thanks to the combination
of energy reduction and satisfaction of QoS requirements.

The contributions of this paper can be summarized as follows:
\begin{itemize}
\item we consider the problem of maximizing operators' profit in green cellular networks;
\item we model the problem as a cooperative game with transferable utility;
\item we devise a distributed algorithm enabling operators to find the coalition maximizing their profits under stability concerns;
\item we show its effectiveness through experimental analysis in realistic scenarios;
\item we assess the impact of energy price and user population on the profits attained by operators.
\end{itemize}

The rest of this paper is organized as follows.
In \dcsSecRef{sys}, we describe the system under study and we present the
problem addressed in this paper.
In \dcsSecRef{game}, we present the cooperative game-theoretic framework we use
to study the problem of coalition formation and the hedonic shift algorithm we
design to form stable coalitions.
In \dcsSecRef{exp}, we show results from an experimental evaluation to show the
effectiveness of the proposed approach.
In \dcsSecRef{related}, we provide an overview of related works.
Finally, in \dcsSecRef{concl}, we conclude the paper and present an outlook
on possible future extensions.

\section{System Model and Problem Definition} \label{sec:sys}

\subsection{System Model} \label{sec:sys-model}

We consider an area served by a set $\mathcal{N}=\{1,\ldots,N\}$ of NOs, 
whose BSs fully cover that area and whose coverage overlaps 
(as typically happens in urban areas~\cite{Niu-2010-Cell,AjmoneMarsan-2012-Multiple,Pollakis-2012-Base}).
To keep the notation simple, we assume that in the
area of interest there is only one BS per NO, so in the rest of this
paper, we will use the terms BS and NO interchangeably
(the extension of the model to support multiple BSs per NO is straightforward).

Each BS $i$ is characterized by its maximum downlink transmission capacity 
$C_i$, and by its power consumption $W_i(n_i)$ that, as
argued in~\cite{Peng-2011-Traffic,Deruyck-2012-Characterization,Lorincz-2012-Measurements},
is linearly dependent on the number of users it is serving, that is:
\begin{equation}\label{eq:bs-powmod}
W_i(n_i) = \alpha_i + \beta_i n_i
\end{equation}
where $\alpha_i$ (the \emph{static term}) is the load-independent power
consumption (which is usually known from the specifications of the BS,
and typically accounts for about 90\% of the total consumption~\cite{Oh-2011-Toward,Peng-2011-Traffic}),
and $\beta_i n_i$
(the \emph{dynamic term}) is the load-dependent power consumption (that
can be determined by linear regression from real power measurements~\cite{Lorincz-2012-Measurements}).

Each NO $i$ provides network connectivity to a set $\mathcal{U}_i$
of \emph{customers} (hereafter also referred to as \emph{users}).
Each user $j\in\mathcal{U}_i$ is characterized by its required QoS, quantified by
the minimum downlink data rate $D_j$ it requests, and
the actual downlink data rate $d_j$ it gets from the network.

Each user $j\in\mathcal{U}$ (where $\mathcal{U}=\bigcup_{i=1}^N\mathcal{U}_i$)
can connect to any BS in the system regardless of
the NO who owns the BS (i.e., it can connect to a BS that belongs to the NO
to which it is subscribed or it can roam on the BS of another one).
This can be accomplished by using techniques like \emph{cell wilting}
and \emph{blossoming}~\cite{Conte-2011-Cell}.
However, the aggregate allocated data rate to users connected to BS $i$ cannot
exceed its capacity $C_i$, that is:
\begin{equation}
\sum_{j\in \mathcal{U}_i}d_j \leq C_i. \label{eq:capacity-constraint}
\end{equation}

We assume that the number of users receiving service from a BS $i$ varies over time,
and is described by the \emph{load profile} curve $\ell_i(t)$ of that BS
that expresses, as function of time, the percentage of the maximum number of users $M_i$
that can receive service by BS $i$ when each user $j$
is allocated its entire desired data rate $D_j$.
It then follows that, if all the users have the same data rate requirement
(i.e., $D_j = D$ for all $j\in\mathcal{U}_i$), then $M_i=C_i/D$.
Conversely, if users are heterogeneous, then $M_j$ is
estimated as $M_j=C_i/\bar{D}$,
where $\bar{D}$ is the weighted average of the data rates requested
by users, i.e., $\bar{D}=\sum_{j\in\mathcal{U}_i} p_j D_j$, where $p_j$ is the
probability that user $j$ arrives at BS $i$.

An example of a typical daily load profile is depicted in
\dcsFigRef{trafload}, where the $x$-axis represents the time (in hours) and the
$y$-axis is the normalized load of the BS \cite{Willkomm-2008-Primary}.
For instance, if at a given time $t$, $M_i=10$ and $\ell_i(t)=0.8$, the number
$n_i(t)$ of user of BS $i$ at time $t$ is $n_i(t)=0.8 \cdot 10 = 8$.
\begin{figure}
\centering
\includegraphics[scale=0.5]{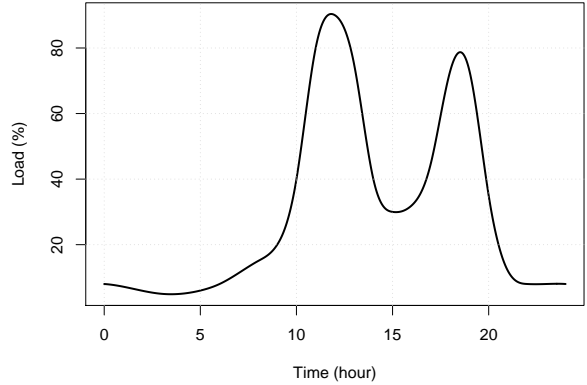}
\caption{A typical daily load profile $\ell_i(\cdot)$ of a single BS $i$.}\label{fig:trafload}
\end{figure}

\subsection{Problem Definition}\label{sec:sys-problem}

Given the system characterized as above and a particular area of interest,
each NO seeks to maximize its \emph{net profit} (i.e., the difference
between its revenues and costs) in the presence of a time-varying population
of users in this area.

The \emph{net profit rate} $P_i$ of NO $i$ (i.e., the profit it makes per unit of
time) can be expressed as
\begin{equation} \label{eq:profit}
P_i = \sum_{j\in \mathcal{U}_i} R_{i,j}-\biggl[W_i(n_i) E_i+\sum_{j\in \mathcal{U}_i} L_{i,j}(d_j)\biggr]
\end{equation}
where $R_{i,j}$ is the revenue rate generated by
user $j$ on BS $i$, $E_i$ is the electricity cost rate
of BS $i$, and $L_{i,j}(d_i)$ is the penalty rate
incurred by NO $i$ if user $j$ receives a downlink rate $d_j$ lower than
its QoS value $D_j$, which is given by the following loss function:
\begin{equation} \label{eq:qos-penalty}
L_{i,j}(d_j) = \biggl(1-\frac{d_j}{D_j}\biggr)R_{i,j}
\end{equation}
Thus, $L_{i,j}(d_i)$ is zero if the QoS of the user is completely satisfied (i.e.,
$d_j=D_j)$, and linearly increases until $R_{i,j}$ as the assigned data rate $d_j$
decreases (so that an NO gets no revenue from those customers that
receive no service).

If the NOs in the area of interest cooperate among them (i.e., they share their
users and BSs) then each NO $i$ can maximize the corresponding value of $P_i$
by acting on the various terms of \dcsEqRef{profit} as follows:
\begin{itemize}
\item it can attempt to reduce $W_i(n_i)$ by offloading (some of) its users
to the BSs of other operators so that its BS can be switched off entirely
(by exploiting \emph{sleep modes}~\cite{Micallef-2010-Cell}) or only
in part (by relying on \emph{cell zooming}~\cite{Niu-2010-Cell});
\item it can attempt to increase $R_{i,j}$ either by attracting users
from other NOs, so that it can better amortize its energy cost $W_i(n_i)$,
or by relying on BSs of other NOs to accept users that, if working alone, it
could not serve without violating \dcsEqRef{capacity-constraint}, thus
incurring into a (possibly high) penalty rate $L_{i,j}$.
\end{itemize} 

It is evident that, to exploit these opportunities, each NO must be willing to cooperate
with (at least some of) the other ones.
In the following section, we will characterize the conditions under which such
a cooperation is not only possible, but also sought by these NOs.

\section{The Coalition Formation Game} \label{sec:game}

As discussed before, cooperation is the key to increase profit.
However, it is unreasonable to assume that a NO is willing
to unconditionally cooperate with the other ones regardless of the benefits
it receives.
As a matter of fact, the acceptance of users roaming from other NOs
is beneficial only if the additional revenue they bring outweighs
the costs and the possible penalties they induce.
Furthermore, the offloading of users to other NOs makes sense only if
a suitable revenue results from this operation for the off-loader.

To cooperate, a set of NOs must first form a \emph{coalition},
i.e., they all must agree to share their own BSs and
users among them.
Given a set of NOs, however, there can be many different coalitions that can be
formed, each one differing from the other ones in terms of the structure (i.e.,
the identity of each member) and/or of the profit it brings to their members.

In order to join a coalition, a NO must indeed find it \emph{profitable},
i.e., it must be sure that the profit it earns by joining the coalition is no worse 
of the one it obtains by working alone.
Furthermore, in order to be sure that this profit is not ephemeral, 
a NO must seek other properties that guarantee the suitability of a coalition,
namely:
\begin{itemize}
\item \emph{Stability}: a coalition is \emph{stable} if none of its participants
finds that it is more profitable to leave it (e.g., to stay alone or to join
another coalition) rather than cooperating with the other ones.
Lack of stability causes possible monetary losses for the following reasons:
\begin{itemize}
\item a NO that has joined a coalition with the expectation of
receiving users roaming from other NOs is penalized if, after switching
on a BS on which to accommodate these users, these NOs leave the coalition;
\item a NO that has accepted more users than those it can serve
without incurring into a penalty, expecting to use the BSs of other
NOs to accommodate them, is penalized if these NOs leave the coalition.
\end{itemize}
\item \emph{Fairness}: when joining a coalition, a NO expects that the resulting
profits are fairly divided among participants.
As an unfair division leads to instability, a fair profit allocation strategy is mandatory.
\end{itemize}
From these considerations, it clearly follows that
a way must be provided to each NO to decide whether
to participate to a coalition or not and, if so, which one
among all the possible coalitions is worth joining.

In this paper, we address this issue by
modeling the problem of coalition formation
as a \emph{coalition formation cooperative game with
transferable utility} \cite{Peleg-2007-CooperativeGames,Book_RAY2007}, where each NO
cooperates with the other ones in order to maximize its net profit rate,
and by devising an algorithm to solve it.
By using our algorithm, the various NOs can make their decisions concerning
coalition membership.

In the rest of this section, we first set the coalition formation problem in the
game-theoretic framework (\dcsSecRef{NO-game}), then we present an algorithm to
form stable coalitions among NOs (\dcsSecRef{algo}), and finally we present an
optimization model to allocate users to a set of BSs, in order to reduce
costs and, at the same time, to meet user QoS for NOs inside the same coalition
(\dcsSecRef{optimization}).

\subsection{Characterization}\label{sec:NO-game}

Our coalition formation algorithm is based on a \emph{hedonic game}~\cite{DrezeGreenberg1980}, a class of
\emph{coalition formation cooperative games}~\cite{Peleg-2007-CooperativeGames,Book_RAY2007}
where each NO acts as a selfish agent and where its preferences over coalitions
depend only on the composition of that coalition.
That is, NOs prefer being in one coalition rather than in another
one solely based on who else is in the coalitions they belong.

Formally, given the set $\mathcal{N}=\{1,2,\ldots,N\}$ of NOs (henceforth also referred to as
the \emph{players}), a \emph{coalition} $\mathcal{S} \subseteq \mathcal{N}$ represents
an agreement among the NOs in $\mathcal{S}$ to act as a single entity.

At any given time, the set of players is partitioned into a
\emph{coalition partition} $\Pi$, that we define as the set
$\Pi = \{\mathcal{S}_1, \mathcal{S}_2, \ldots , \mathcal{S}_l \}$,
where $\mathcal{S}_k \subseteq \mathcal{N}$ ($k=1,\ldots, l$)  is a
disjoint coalition such that $\bigcup_{k=1}^l \mathcal{S}_k = \mathcal{N}$ and
$\mathcal{S}_j\cap\mathcal{S}_k=\emptyset$ for $j\ne k$.
Given a coalition partition $\Pi$, for any NO $i \in \mathcal{N}$, we denote as
$\mathcal{S}_{\Pi}(i)$ the coalition to which $i$ is participating.

Each coalition $\mathcal{S}$ is associated with its \emph{coalition value} $v(\mathcal{S})$,
that we define as the \emph{net profit rate} of that coalition, that is:
\begin{equation} \label{eq:value}
v\bigl(\mathcal{S}\bigr) = R\bigl(\mathcal{U}_\mathcal{S}\bigr) - Q\bigl(\mathcal{U}_{\mathcal{S}}\bigr) - K\bigl(\mathcal{S}\bigr)
\end{equation}
where:
\begin{itemize}
\item $R\bigl(\mathcal{U}_{\mathcal{S}}\bigr)$ is the \emph{coalition revenue rate}, corresponding to the sum of revenue rates of
individual users $j \in \mathcal{U}_{\mathcal{S}}$ (where
$\mathcal{U}_{\mathcal{S}} = \bigcup_{i\in\mathcal{S}}\mathcal{U}_i$
is the joint user population of the NOs belonging to $\mathcal{S}$);
\item $Q\bigl(\mathcal{U}_{\mathcal{S}}\bigr)$ is the \emph{coalition load cost rate}, and is computed by
minimizing the costs resulting from serving the users in
$\mathcal{U}_{\mathcal{S}}$ using all the resources provided by
the NOs belonging to $\mathcal{S}$ (we discuss this in \dcsSecRef{optimization});
\item $K(\mathcal{S})$ is the \emph{coalition formation cost rate}, that takes into
account the cost incurred by players to establish and maintain the coalition
(e.g., the costs for system reconfiguration to enable user migration and handover across
NOs).
In this paper, we assume $K(\mathcal{S})$ to be proportional to the coalition
size, and we define it as:

\begin{equation}
K(\mathcal{S})=\begin{cases}
                  \sum_{i \in \mathcal{S}}{K_i}, & \lvert\mathcal{S}\rvert > 1,\\
                  0, & \text{otherwise}.
                 \end{cases}
\end{equation}
where $K_i$ is the coalition formation cost rate for NO $i$.
\end{itemize}

Obviously, each NO $i\in\mathcal{S}$ must receive a fraction $x_i(\mathcal{S})$
of the coalition value, that we call the \emph{payoff} of $i$ in $\mathcal{S}$.
Our game is conceived in such a way to form coalitions in which NOs get payoffs as
high as possible, without violating the fairness requirement,
so that stability is achieved.
Thus, a \emph{payoff allocation rule} must be specified in order to compute the
payoffs of each coalition member in such a way to ensure fairness in the division of payoffs.

To this end, we use the \emph{Shapley value}~\cite{SHAPLEY_53}, a payoff
allocation rule that is based on the concept of  \emph{marginal
contribution} of players (i.e., the change in the worth of a coalition when a player
joins to that coalition), such that the larger is the contribution provided by a
player to a coalition, the higher is the payoff allocated to it.~\footnote{More specifically, we use the \emph{Aumann-Dr\'eze} value~\cite{Aumann-1974-Cooperative}, which is an extension of the Shapley value for games with coalition structures.}
This means that, in a given coalition, some ``more-contributing'' NOs will be
rewarded by other ``less-contributing'' NOs to encourage them to join the coalition.
More specifically, the Shapley value $\phi_i(v)$ of player $i$ is defined as:
\begin{equation} \label{eq:shapley}
\phi_i\bigl(v\bigr) = \sum_{\mathcal{S} \subseteq \mathcal{N} \setminus \{i\}}
\frac{\lvert\mathcal{S}\rvert!\bigl(N-\lvert\mathcal{S}\rvert-1\bigr)!}{N!} \Bigl( v\bigl(\mathcal{S} \cup \{i\}\bigr)-v\bigl(\mathcal{S}\bigr)\Bigr)
\end{equation}
where the sum is over all subsets $\mathcal{S}$ not containing $i$
(the symbol ``$\setminus$'' denotes the set difference operator), and the symbol
``$!$'' denotes the factorial function.

It is worth noting that we rely on the Shapley value for its interesting properties.
Nevertheless, other payoff allocation rules can be used and our work is general enough to support them.

To set up the coalition formation process, we need to define, for each
NO $i$, a \emph{preference relation} $\succeq_i$ that NO $i$ can use to
order and compare all the possible coalitions it may join.
Formally, this corresponds to define a complete, reflexive, and transitive binary relation
over the set of all coalitions that NO $i$ can form (see~\cite{BOGOMONLAIA-JACKSON_2002}).

Specifically, for any NO $i\in\mathcal{N}$ and given $\mathcal{S}_1,\mathcal{S}_2\subseteq\mathcal{N}$,
the notation $\mathcal{S}_1 \succeq_i \mathcal{S}_2$ means that NO $i$
prefers being a member of $\mathcal{S}_1$ over $\mathcal{S}_2$ or at least $i$ prefers
both coalitions equally.
In our coalition formation game, for any NO $i \in \mathcal{N}$, we use the
following preference relation:
\begin{equation}
\label{eq:preferences}
\mathcal{S}_1 \succeq_i \mathcal{S}_2 \Leftrightarrow u_i(\mathcal{S}_1) \ge u_i(\mathcal{S}_2),
\end{equation}
where $\mathcal{S}_1,\mathcal{S}_2 \subseteq \mathcal{N}$ are any two coalitions that
contain NO $i$ (i.e., $i \in \mathcal{S}_1$ and $i \in \mathcal{S}_2$),
and $u_i$
is a preference function defined
for any NO $i$ as follows:
\begin{equation}
\label{eq:preferencesUTIL}
u_i(\mathcal{S}) = \begin{cases}
                 x_i(\mathcal{S}), & \mathcal{S} \notin h(i), \\
                 -\infty, & \text{otherwise}.
                \end{cases}
\end{equation}
where $x_i(\mathcal{S})$ is the payoff received by NO $i$ in $\mathcal{S}$ by
means of \dcsEqRef{shapley}, and $h(i)$ is a \emph{history} set where NO $i$ stores the
identity of the coalitions that have been already evaluated so that
we avoid generating twice the same candidate coalition
(a similar idea for pruning already considered coalitions has also been used in previously
published work, such as in \cite{Saad-2011-Hedonic}).

Thus, according to \dcsEqRef{preferencesUTIL}, each NO prefers to join the coalition
that provides the larger payoff. 

The strict counterpart of $\succeq_i$, denoted by $\succ_i$,
is defined by replacing $\ge$ with $>$ in \dcsEqRef{preferences},
and implies that $i$ strictly prefers being a member of $\mathcal{S}_1$ over $\mathcal{S}_2$.

\subsection{The Algorithm for Coalition Formation} \label{sec:algo}

In this section, we present an algorithm for coalition
formation that allows the NOs to take distributed decisions
for selecting which coalitions to join at any point in time.
This algorithm is based on the following \emph{hedonic shift rule}
(see \cite{Saad-2011-Hedonic}):
\begin{definition}
Given a coalition partition
$\Pi=\{\mathcal{S}_1, \ldots, \mathcal{S}_h \}$ on the set $\mathcal{N}$ and a preference
relation $\succ_i$, any NO $i \in \mathcal{N}$ decides to leave its current
coalition $\mathcal{S}_{\Pi}(i)=\mathcal{S}_l$, for $1 \leq l \leq h$,
to join another one $\mathcal{S}_k \in \Pi \cup \emptyset$, with $\mathcal{S}_k \neq \mathcal{S}_l$,
if and only if
$\mathcal{S}_k \cup \{i \} \succ_i \mathcal{S}_{l}$,
that is if its payoff in the new coalition exceeds the one
it is getting in its current coalition.
Hence, $\{ \mathcal{S}_l, \mathcal{S}_k \} \to \{ \mathcal{S}_l \backslash \{i\} , \mathcal{S}_k \cup \{i\} \}$.
\end{definition}
This shift rule (that we denote as ``$\to$'') provides a mechanism through
which any NO can leave its current coalition $\mathcal{S}_{\Pi}(i)$
and join another coalition $\mathcal{S}_k$, given that the new
coalition $\mathcal{S}_k \cup \{i\}$ is strictly preferred over $\mathcal{S}_{\Pi}(i)$ through
any preference relation that the NOs are using.
This rule can be seen as a selfish decision made by a NO to move from its
current coalition to a new one, \emph{regardless of the effects of this move on
the other NOs}.

Using the hedonic shift rule, we design a distributed hedonic coalition
formation algorithm for NOs as presented in \dcsAlgRef{hcf}.
\begin{algorithm}
\caption{The Coalition Formation Algorithm for NOs}\label{alg:hcf}
\begin{algorithmic}[1]
\Procedure{CoalitionFormation}{$\mathit{state},i$}
 \State $h \gets \emptyset$ \label{alg:hcf-init-beg}
 \State $\mathcal{S}_\mathrm{best} \gets \emptyset$ \label{alg:hcf-init-end}
 \Repeat
  \State \Call{Lock}{$\mathit{state}$} \label{alg:hcf-lock}
  \State $\Pi_c \gets \Call{GetCurrentPartition}{\mathit{state}}$
  \State $\mathcal{S}_\mathrm{cur} \gets \mathcal{S}_{\Pi_c}(i)$
  \State $\mathcal{S}_\mathrm{best} \gets \mathcal{S}_\mathrm{cur}$
  \ForAll{$\mathcal{S} \in \bigl(\Pi_c \setminus \{\mathcal{S}_\mathrm{cur}\}\bigr) \cup \emptyset \text{ \textbf{and} } \mathcal{S} \notin h$} \label{alg:hcf-stage1-beg}
   \State $\mathcal{S}_\mathrm{new} \gets \mathcal{S} \cup \{i\}$
   \State $x_{\mathcal{S}_\mathrm{best}} \gets \Call{ComputePayoff}{\mathcal{S}_\mathrm{best},i}$ \Comment{See \dcsEqRef{value} and \dcsEqRef{shapley}}
   \State $x_{\mathcal{S}_\mathrm{new}} \gets \Call{ComputePayoff}{\mathcal{S}_\mathrm{new},i}$ \Comment{See \dcsEqRef{value} and \dcsEqRef{shapley}}
   \If{$x_{\mathcal{S}_\mathrm{new}} > x_{\mathcal{S}_\mathrm{best}}$} \Comment{See \dcsEqRef{preferences} and \dcsEqRef{preferencesUTIL}}
    \State $\mathcal{S}_\mathrm{best} \gets \mathcal{S}_\mathrm{new}$
   \EndIf
  \EndFor \label{alg:hcf-stage1-end}
  \If{$\mathcal{S}_\mathrm{best} \ne \mathcal{S}_\mathrm{cur}$} \label{alg:hcf-stage2-beg}
   \State $\mathcal{S} \gets \mathcal{S}_\mathrm{cur} \setminus \{i\}$
   \State $\mathcal{T} \gets \mathcal{S}_\mathrm{best} \setminus \{i\}$
   \State $\Call{UpdateHistory}{h,\mathcal{S}}$
   \State $\Pi_\mathrm{best} \gets \bigl(\Pi_c \setminus \{\mathcal{S}_\mathrm{cur},\mathcal{T}\}\bigr) \cup \bigl\{\mathcal{S},\mathcal{S}_\mathrm{best}\bigr\}$
   \State \Call{SetCurrentPartition}{$\mathit{state},\Pi_\mathrm{best}$}
  \EndIf \label{alg:hcf-stage2-end}
  \State \Call{Unlock}{$\mathit{state}$} \label{alg:hcf-unlock}
 \Until{$\mathcal{S}_\mathrm{best} = \mathcal{S}_\mathrm{cur}$} \label{alg:hcf-until}
\EndProcedure
\end{algorithmic}
\end{algorithm}

The basic idea of the algorithm is to have each NO $i$ search, asynchronously
with respect to the other NOs, the state space of possible coalitions it may
join, and for each one of them, evaluate whether it is preferable
(according to the corresponding $\succ_i$ relation) to remain in its current
coalition, or to join it.
Whenever a NO decides to move from a coalition to another one, it updates its
history set $h(i)$ by appending the coalition it is leaving, so that the same
coalition is not visited twice during the coalition space search.
A NO iterates the actions listed in \dcsAlgRef{hcf} until no more hedonic shift
rules are possible.
It is worth noting that the asynchronicity of our algorithm makes it suitable
to be executed, for instance, when new users arrive to NOs, thus making it
able to adapt to environmental changes.

Let us explain in detail how \dcsAlgRef{hcf} works.
The algorithm takes as parameters the global state $\mathit{state}$, storing the
current shared coalition partition $\Pi_c$, and the identity $i$ of the calling
NO (initially there are no coalitions, i.e.,
$\Pi_c = \Pi_0 = \bigl\{ \{ 1 \}, \{ 2 \}, \ldots, \{ N \} \bigr\}$).

At each execution of the algorithm, NO $i$ initializes its history set $h$ and other auxiliary variables (lines \ref{alg:hcf-init-beg}--\ref{alg:hcf-init-end}), and then enters
a loop that is executed until no more hedonic shift rules can be performed
from the last coalition partition considered by $i$.

In each loop iteration, NO $i$ retrieves the current coalition partition,
and generates all the possible hedonic shifts until no more of them are possible.
Given the distributed nature of the algorithm, we postulate the 
use of suitable distributed space management algorithms (e.g. \cite{Kshemkalyani-2008-Distributed,Weiss-2013-MAS}).

Then, after acquiring a lock to gain exclusive access to the shared state (line
\ref{alg:hcf-lock}) in order to ensure its atomic update 
(by means of a suitable distributed  mutual exclusion algorithm \cite{Kshemkalyani-2008-Distributed}),
NO $i$ iteratively evaluates all the possible coalitions it
can form from its current coalition partition, to look for the one with the
higher payoff.

To do so, given its current coalition partition $\Pi_c$, for each coalition
$\mathcal{S}_k\in\Pi_c\cup\emptyset$ (not present in its history set and
different from its current one $\mathcal{S}_{\Pi_c}(i)$), NO $i$ applies the
hedonic shift rule and evaluates its preference against the current coalition
$\mathcal{S}_{\Pi_c}(i)$ (lines \ref{alg:hcf-stage1-beg}--\ref{alg:hcf-stage1-end}).

If a coalition $\mathcal{S}_k$ with the higher payoff is found
(lines \ref{alg:hcf-stage2-beg}--\ref{alg:hcf-stage2-end}), NO $i$ adds to its
history set $h$ the coalition $\mathcal{S}_{\Pi_c}(i)\setminus\{i\}$ it is
leaving, and updates the partition set by updating both $\mathcal{S}_k$
(that now contains also $i$) and $\mathcal{S}_{\Pi_c}(i)$ (that now does not
contain $i$ anymore).

Then, after releasing the exclusive lock to the shared state (line
\ref{alg:hcf-unlock}), NO $i$ repeats the above steps (lines
\ref{alg:hcf-lock}--\ref{alg:hcf-unlock}) to look for a better
coalition, in case some other NO has meanwhile modified 
the shared state by changing the coalition partition.

Eventually, if no other better coalition is found, NO $i$ terminates the
execution of the algorithm (line \ref{alg:hcf-until}), until a new instance
is run again.

The convergence of the proposed algorithm during the hedonic
coalition formation phase is guaranteed as follows:
\begin{proposition}[Convergence] \label{thm:convergence}
Starting from any initial coalition structure $\Pi_0$, the proposed algorithm
always converges to a final partition $\Pi_f$.
\end{proposition}
\begin{proof}
The coalition formation phase can be mapped to a sequence of shift operations.
That is, according to the hedonic shift rule, every shift operation transforms
the current partition $\Pi_c$ into another partition $\Pi_{c+1}$.
Thus, starting from the initial step, the algorithm yields the following
transformations:
\begin{equation} \label{eq:sequence}
\Pi_0 \to \Pi_1 \to \cdots \to \Pi_c \to \Pi_{c+1}
\end{equation}
where the symbol ``$\to$'' denotes the application of a shift operation.
Every application of the shift rule leads to a coalition partition that has not
been previously visited (i.e., a new coalition partition).
Thus, the number of transformations performed by the shift rule is finite (at
most, it is equal to the number of partitions, that is the Bell number) and
hence the sequence in \dcsEqRef{sequence} will always terminate and converge to
a final partition $\Pi_f$.
\end{proof}
The stability of the final partition $\Pi_f$ resulting from the
convergence of the proposed algorithm can be addressed by using 
the following stability definition (see \cite{BOGOMONLAIA-JACKSON_2002} for details).

\begin{definition}
\label{NASH_STAB}
{\em
A coalition partition $\Pi=\{\mathcal{S}_1, \ldots, \mathcal{S}_l\}$ is \emph{Nash-stable}
if $\forall i\in\mathcal{N}$, $\mathcal{S}_{\Pi}(i)\succeq_i\mathcal{S}_k\cup\{i\}$
for all $\mathcal{S}_k \in \Pi \cup \emptyset$.
}
\end{definition}
It is worth noting that Nash-stability captures the notion of stability
with respect to movements of single NOs (i.e., no NO has an incentive to
unilaterally deviate).

For the hedonic coalition formation phase of the proposed
algorithm, we can prove the following result: 

\begin{proposition}[Nash-stability] \label{thm:nash-stable}
Any final partition $\Pi_f$ resulting from \dcsAlgRef{hcf} is Nash-stable.
\end{proposition}
\begin{proof}
We prove it by contradiction.
Assume that the final partition $\Pi_f$ is not Nash-stable.
Consequently, there exists a NO $i \in \mathcal{N}$ and a coalition
$\mathcal{S}_k\in\Pi_f\cup\emptyset$ such that $\mathcal{S}_k\cup\{i\}\succ_i\mathcal{S}_{\Pi_f}(i)$.
Then, NO $i$ will perform a hedonic shift operation and hence
$\Pi_f \to \Pi'_f$, where $\Pi_f'$ is the new coalition partition
resulting after the hedonic shift operation.
This contradicts the assumption that $\Pi_f$ is the final outcome of our
algorithm.
\end{proof}

It is worth to point out that Nash-stability also implies the so called \emph{individual-stability}~\cite{BOGOMONLAIA-JACKSON_2002}.
A partition $\Pi=\{ \mathcal{S}_1,\ldots,\mathcal{S}_l \}$ is
\emph{individually-stable} if it does not exist a NO $i \in \mathcal{N}$ and a
coalition $\mathcal{S}_k \in \Pi \cup \emptyset$ such that
$\mathcal{S}_k\cup\{i\}\succ_i\mathcal{S}_{\Pi}(i)$ and
$\mathcal{S}_k\cup\{i\}\succeq_j\mathcal{S}_k$ for all $j \in \mathcal{S}_k$,
i.e., if no NO can benefit by moving from its coalition to another
existing (possibly empty) coalition while not making the members of that
coalition worse of.
Thus, we can conclude that our algorithm always converges to a partition $\Pi_f$
which is both Nash-stable and individually stable.

Given the NP-completeness of the problem of finding a Nash-stable partition
\cite{Ballester-2004-NP}, the computational cost of our algorithm can become
quite large when the number of NOs increases
(indeed, in the worst case, it is bounded by the $N^\text{th}$ Bell number,
where $N$ is the number of NOs).
In these cases, we can reduce the computational cost by having each NO check
for the Nash-stability of its current coalition partition (this check takes
polynomial time~\cite{Ballester-2004-NP}), and execute the algorithm only
if Nash-stability no longer holds true.

\subsection{Computation of the Optimal Coalition Load Cost}\label{sec:optimization}

The algorithm presented in the previous section requires the computation of the
value $v(\mathcal{S})$ of any coalition $\mathcal{S}$ that each NO $i$ may possibly join,
that in turn requires the computations of coalition load cost rate $Q(\mathcal{U}_{\mathcal{S}})$
(see \dcsEqRef{value}).
To compute $Q(\mathcal{U}_{\mathcal{S}})$, we need in turn to determine, for
the coalition $\mathcal{S}$, the optimal data rate allocation (i.e., the
allocation of users that minimizes the costs of NOs).

To this end, we define a \emph{Mixed Integer Linear Program} (MILP)
modeling the problem of allocating a set $\mathcal{U}_{\mathcal{S}}$ of users
onto a set $\mathcal{S}$ of BSs so that the overall cost rates
of NOs in $\mathcal{S}$ are minimized.
The resulting optimization model is shown in \dcsFigRef{opt-maxprofit}, where
we use the same notation introduced in \dcsSecRef{sys} (however, to ease
readability, we denote with $\mathcal{U}$ the user set, i.e., we drop
the dependence from $\mathcal{S}$).
\begin{figure}
\centering
\begin{small}
\dcsHRule
\begin{subequations}\label{eq:opt-maxprofit}
\begin{align}
\text{minimize} & \,\, Q(\mathcal{U}) = \sum_{i \in \mathcal{S}}{\Biggl[b_i W_i\biggl(\sum_{j \in \mathcal{U}}{u_{i,j}}\biggr) E_i + \sum_{j \in \mathcal{U}}{L_{i,j}\bigl(d_{i,j}\bigr)}\Biggr]} \label{eq:opt-maxprofit-obj}\\
\text{subject to} \nonumber \\
& \sum_{j \in \mathcal{U}}{d_{i,j}} \le C_i, \qquad \qquad \qquad \qquad \; \; \; i \in \mathcal{S}, \label{eq:opt-maxprofit-c1}\\
& \sum_{i \in \mathcal{S}}{u_{i,j}} = 1, \qquad \qquad \qquad \qquad \; \; \; \; j \in \mathcal{U}, \label{eq:opt-maxprofit-c2}\\
& \sum_{j \in \mathcal{U}}{u_{i,j}} \le b_i U, \qquad \qquad \qquad \qquad \; i \in \mathcal{S}, \label{eq:opt-maxprofit-c3}\\
& d_{i,j} \le u_{i,j}D_j, \qquad \qquad \qquad \qquad \; \; \, i \in \mathcal{S}, j \in \mathcal{U}, \label{eq:opt-maxprofit-c4}\\
& d_{i,j} \in \mathbb{R}^*, \qquad \qquad \qquad \qquad \; \; \; \; \; \; \; \; i \in \mathcal{S}, j \in \mathcal{U}, \label{eq:opt-maxprofit-c6}\\
& u_{i,j} \in \lbrace 0,1 \rbrace, \qquad \qquad \qquad \qquad \; \; \; i \in \mathcal{S}, j \in \mathcal{U}, \label{eq:opt-maxprofit-c7}\\
& b_i \in \lbrace 0,1 \rbrace, \qquad \qquad \qquad \qquad \; \; \; \; \; i \in \mathcal{S}. \label{eq:opt-maxprofit-c8}
\end{align}
\end{subequations}
\dcsHRule
\end{small}
\caption{The user-to-BS allocation optimization model.}\label{fig:opt-maxprofit}
\end{figure}

In the optimization model we use the following decision variables:
\begin{itemize}
\item $u_{i,j}$, which is a binary
variable that is equal to $1$ if user $j$ is allocated to BS $i$;
\item $d_{i,j}$, which is a real variable representing the downlink data rate allocated to
user $j$ by BS $i$;
\item $b_i$, which is a binary variable that is equal to $1$ if BS $i$ is
switched on.
\end{itemize}
The objective function $Q\bigl(\mathcal{U}\bigr)$ (see \dcsEqRef{opt-maxprofit-obj})
represents the cost rates incurred by the coalition of NOs for serving users in $\mathcal{U}$,
and is defined as the sum of the costs due to the power absorbed by
the BSs that are switched-on, and of those due to QoS violations (if any).

The resulting optimal user allocation is bound to the following constraints:
\begin{itemize}
\item \dcsEqRef{opt-maxprofit-c1} ensures that the capacity of a switched-on BS
is not exceeded;
\item \dcsEqRef{opt-maxprofit-c2} imposes that each user is served by exactly
one BS;
\item \dcsEqRef{opt-maxprofit-c3} states that only BSs that are switched on
can serve users; the purpose of this constraint is to avoid that a user is
served by a BS that will be switched off;
\item \dcsEqRef{opt-maxprofit-c4} imposes that each user obtains at most the
requested data rate by the serving BS;
\item \dcsEqRef{opt-maxprofit-c6}, \dcsEqRef{opt-maxprofit-c7}, and
\dcsEqRef{opt-maxprofit-c8} define the domain of decision variables
$d_{i,j}$, $u_{i,j}$, and $b_i$, respectively.
\end{itemize}

As can be noted from the definition of $Q\bigl(\mathcal{U}\bigr)$,
the solution of the optimization problem at a specif instant of time requires
the knowledge of the number of users present in each BS $i$ at that time.
In general, however, such number is not constant, but it varies over
time according to the corresponding load profile $\ell_i(t)$. 

In order to compute the number of users of BS $i$ at time $t$ from $\ell_i$, we proceed as follows:
first, as typically done in the literature~\cite{Lorincz-2012-Impact,Boiardi-2013-Radio},
we discretize  $\ell_i(t)$ by splitting the time axis into uniform disjoint
sub-intervals
$[\tau,\tau+\Delta t)$ of length $\Delta t$ time units (where $\Delta t$ is the
\emph{discretization step}).
Then, we approximate the (normalized) load of each
subinterval as a constant value set to the peak load of that subinterval.
For instance, the result of the discretization of the load profile of
\dcsFigRef{trafload} with $\Delta t$ of $1$ hour is depicted in
\dcsFigRef{trafload-discrete}.
In this figure, the time-horizon of one day (i.e., $24$ hours) is split into
several subintervals $[\tau,\tau+\Delta t)$ of length $\Delta t=1$ hour, where
$\tau=0,1,\ldots,23$.
Each subinterval is delimited by vertical dotted segments, while every
horizontal solid red segment is the peak inside each subinterval, that we will
use as an approximation of the (normalized) load inside the subinterval.

\begin{figure}
\centering
\includegraphics[scale=0.5]{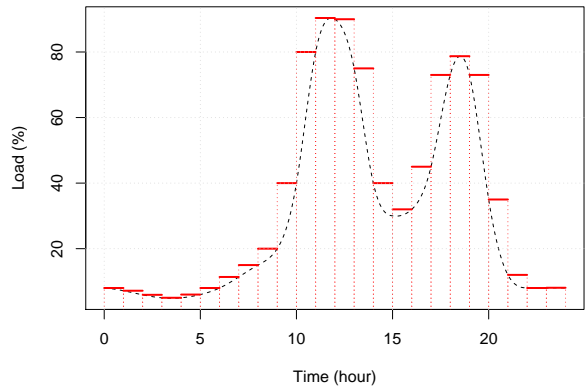}
\caption{Discretization of the load profile of \dcsFigRef{trafload} with a time-horizon of $1$ day and $\Delta t=1$ hour (vertical segments represent subintervals bounds and horizontal segments are peaks inside subintervals).}
\label{fig:trafload-discrete}
\end{figure}

\section{Experimental Evaluation} \label{sec:exp}

In order to assess the ability of our algorithm of increasing the net profits for a population of NOs,
we perform a set of experiments in which we consider a variety of 
realistic scenarios and real-world traffic data.
In these experiments we vary, in a controlled way, various input parameters
of the algorithm, namely the cost of energy, the QoS requirements of users, and
the discretization step of the traffic profile curve, so that we are able to
assess the impact of each one of them on the performance of the algorithm.
The results we collect, discussed in this section, demonstrate
the ability of our algorithm of yielding significant increases of
the net profit achieved by a set of NOs in all the scenarios we consider.

To perform such experiments, we develop an ad-hoc simulator written in
\texttt{C++} and interfaced with CPLEX \cite{CPLEX} to solve the various
instances of the optimization model presented in \dcsSecRef{optimization}.

\subsection{Experimental Setup} \label{sec:exp-cfg}

We consider a system configuration comprising five NOs,
each one owning a single BS. Without loss of generality,
we assume that all the BSs are identical in terms of
capacity and energy consumption. More specifically,
we set $C_i=100$ Mbps, 
$\alpha_i=0.551$ kW, and $\beta_i=0.00146$ kW for $i=\{1,\ldots,5\}$
(these last two values have been taken from~\cite{Lorincz-2012-Measurements}).
We also assume that all NOs incur in the same coalition cost rate,
i.e., $K_i = 0.01$ \$/hour $\forall i\in \{1,\ldots,5\}$.

Furthermore, we assume that each BS has its own load profile, that differs from
those of the other ones.
The load profiles we consider in our experiments, 
reported in Figs.~\ref{fig:exp-load_0}--\ref{fig:exp-load_4},
have been obtained from real-world data~\cite{ANRG_Data} consisting of 
normalized cellular traffic collected, with a resolution of $30$ minutes, in a
metropolitan urban area during one week, and been already used in similar
studies~\cite{AjmoneMarsan-2012-Multiple,Son-2011-Base}.
Specifically, each load profile curve $\ell_i(\cdot)$ of \dcsFigRef{exp-load}
has been obtained by fitting a periodic cubic spline to the traffic data related
to BS $i$.
\begin{figure}
\centering
\subfloat[][Traffic load curve for BS $1$\label{fig:exp-load_0}]{
\centering
\includegraphics[scale=.30]{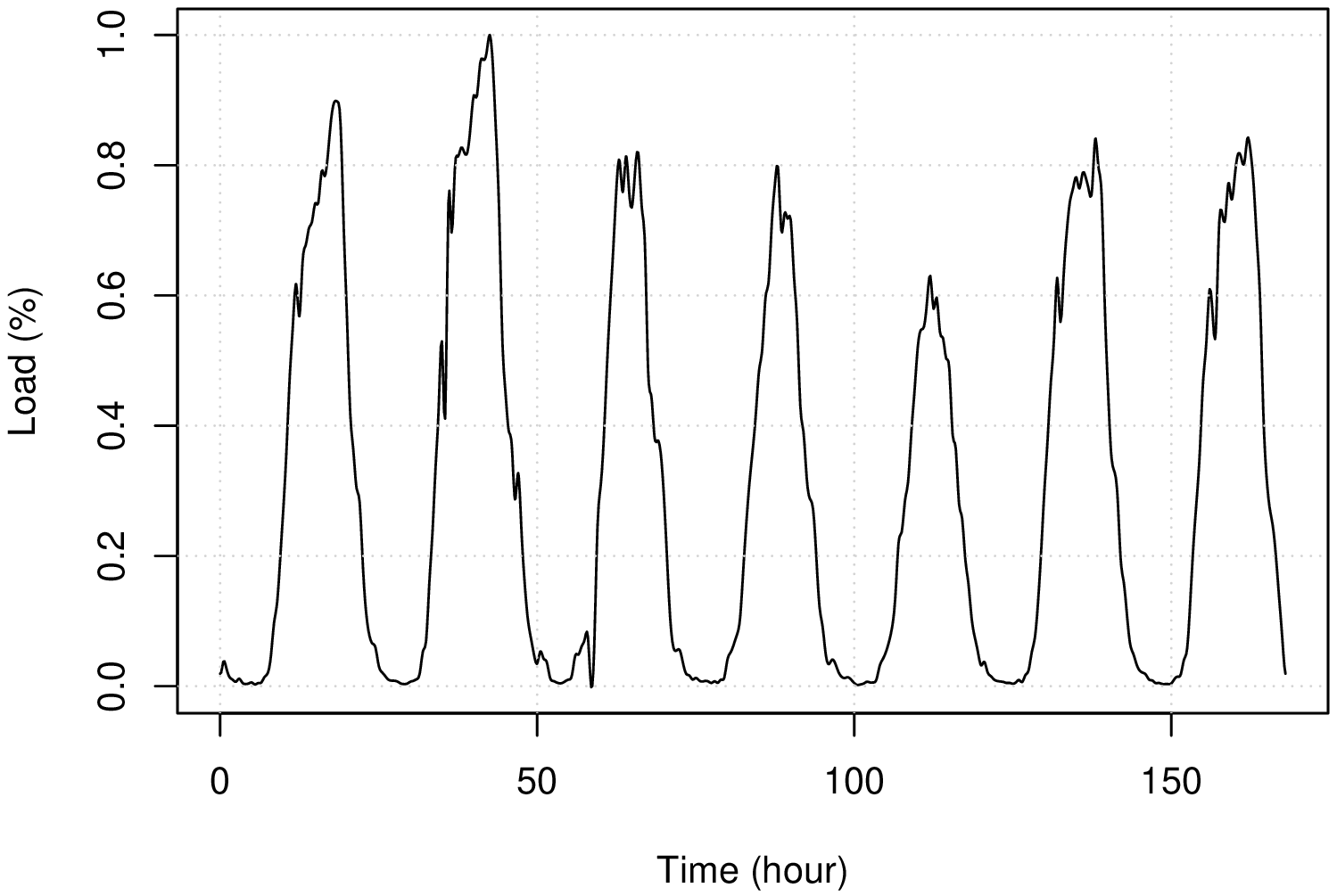}
}
~
\subfloat[][Traffic load curve for BS $2$\label{fig:exp-load_1}]{
\centering
\includegraphics[scale=.30]{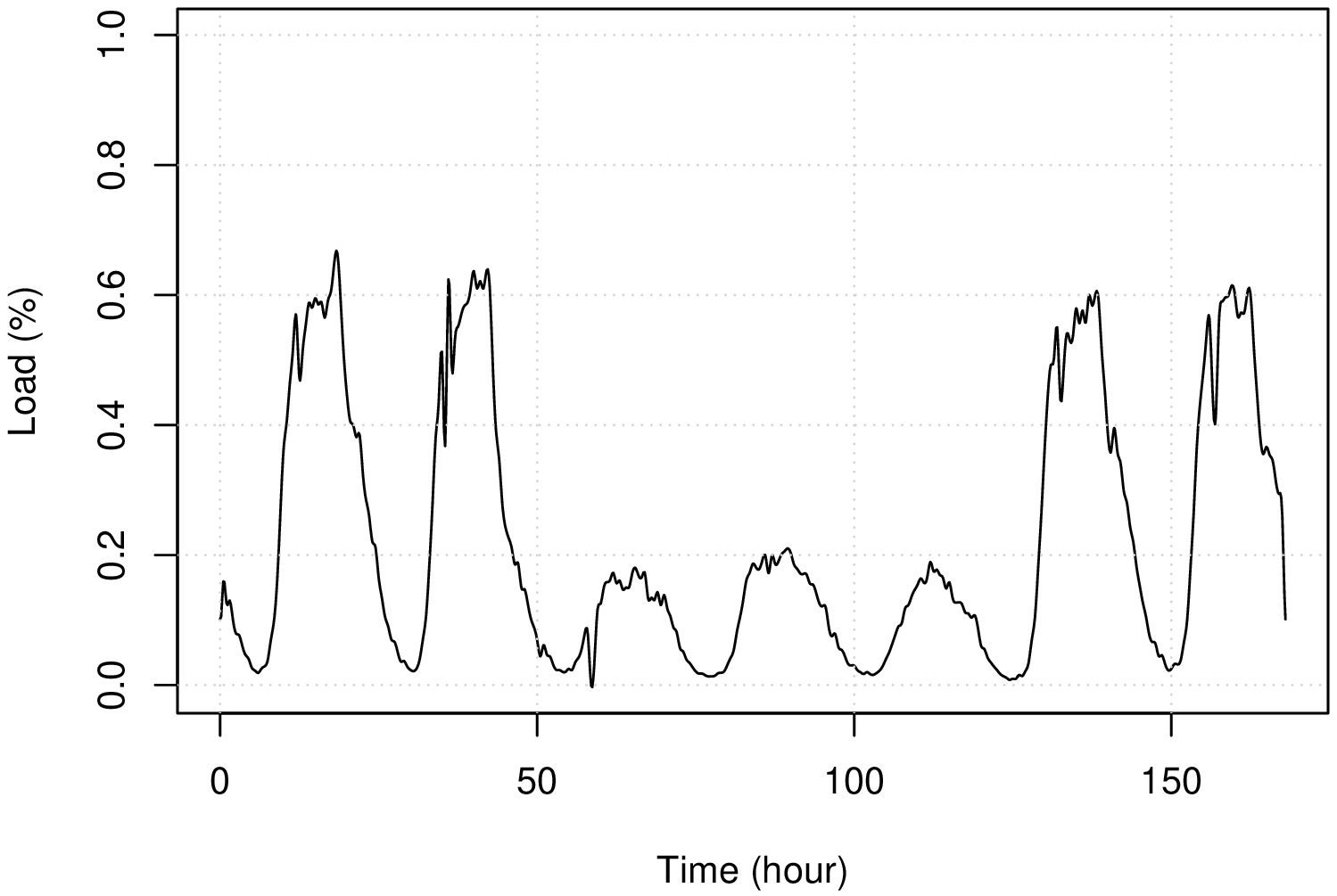}
}
\\
\subfloat[][Traffic load curve for BS $3$\label{fig:exp-load_2}]{
\centering
\includegraphics[scale=.30]{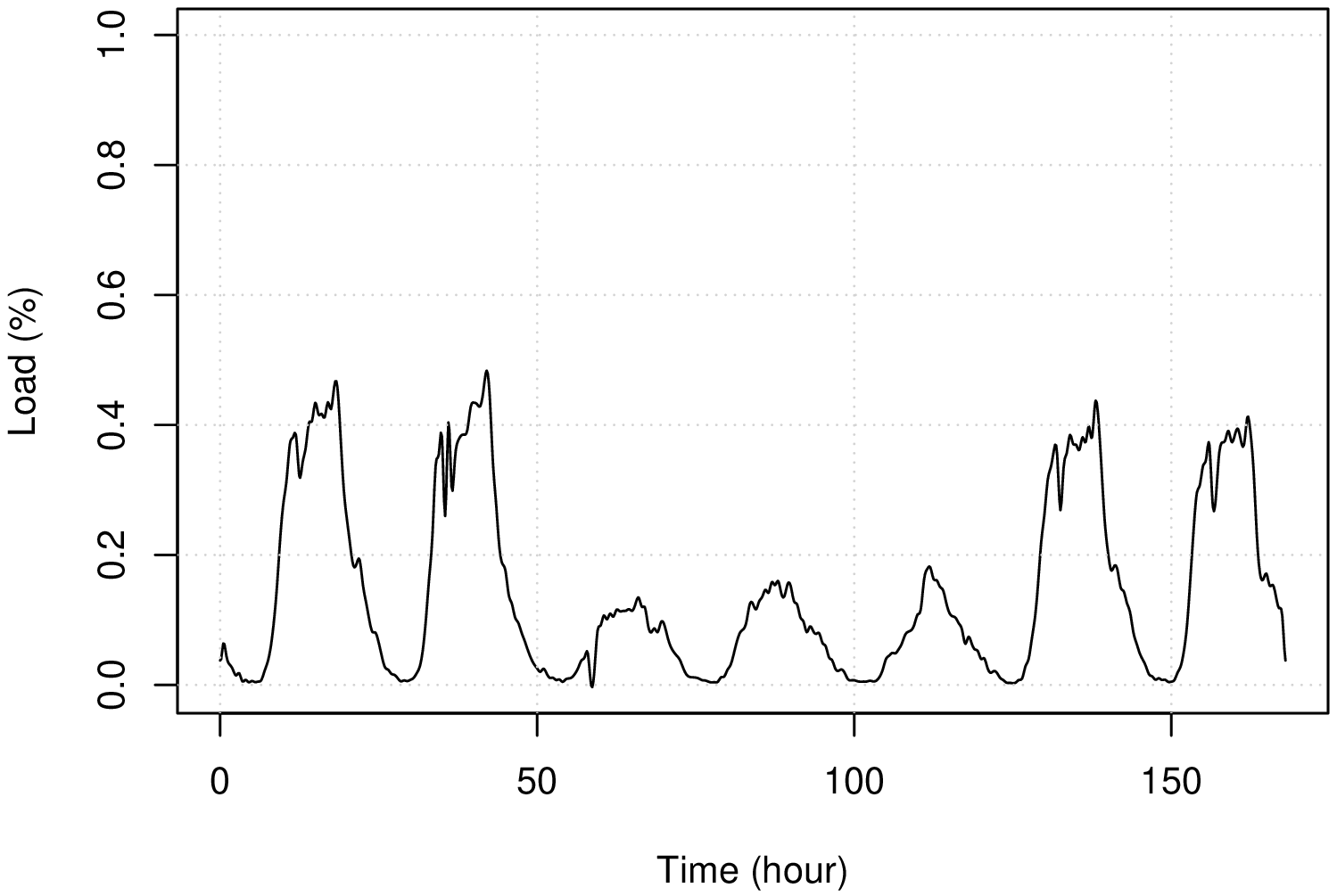}
}
~
\subfloat[][Traffic load curve for BS $4$\label{fig:exp-load_3}]{
\centering
\includegraphics[scale=.30]{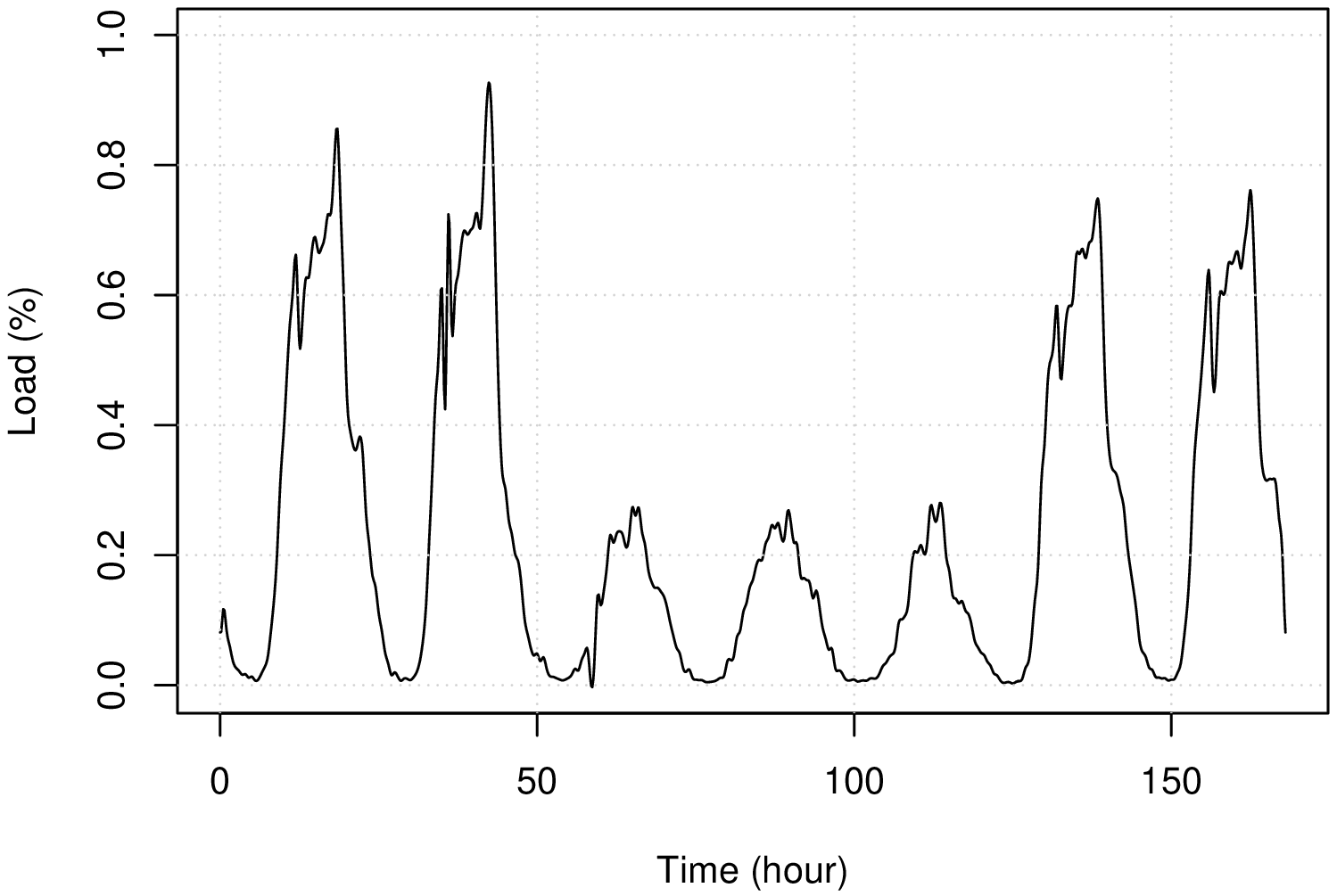}
}
\\
\subfloat[][Traffic load curve for BS $5$\label{fig:exp-load_4}]{
\centering
\includegraphics[scale=.30]{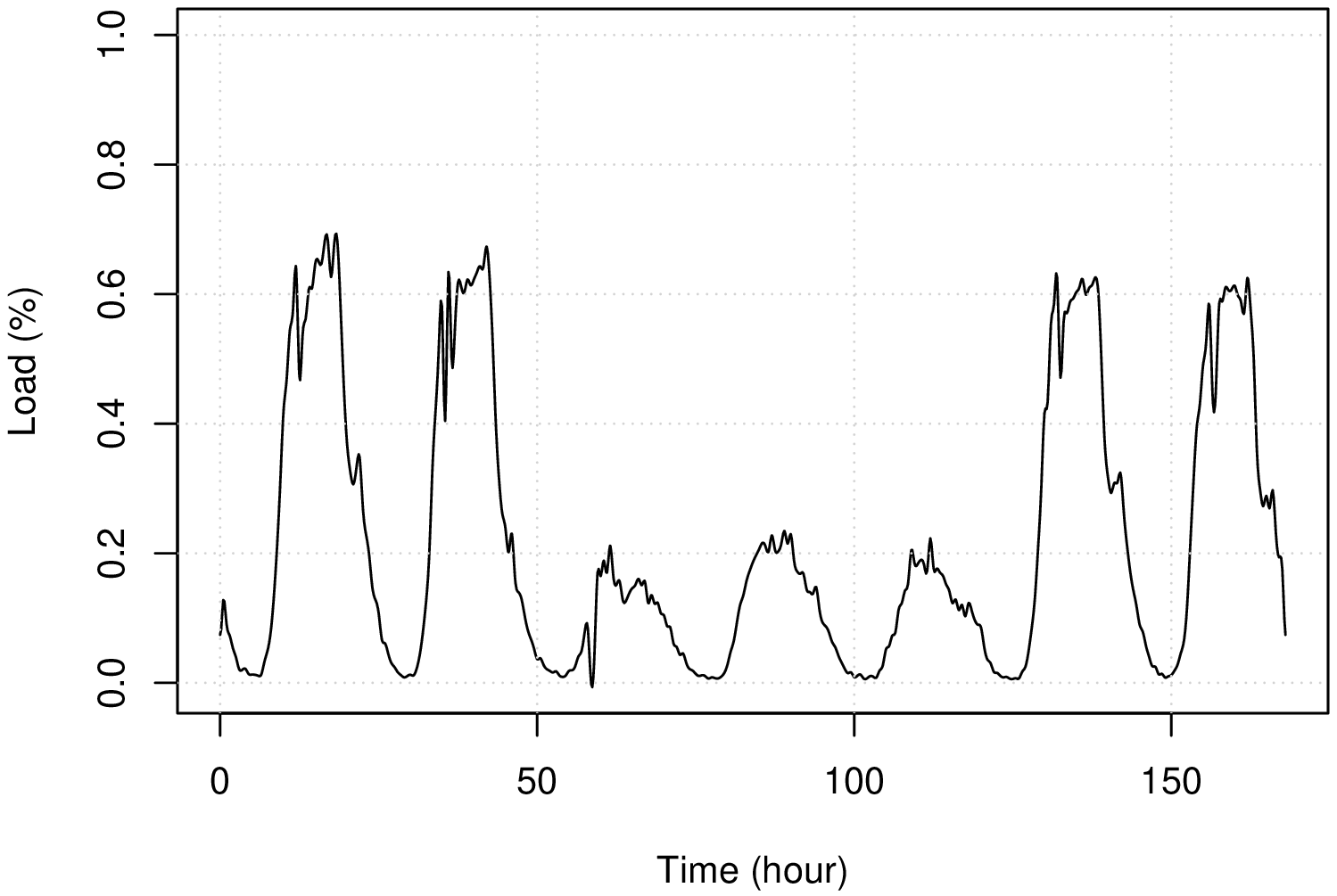}
}
\caption{Traffic load curves for the BSs of the experimental scenarios.}\label{fig:exp-load}
\end{figure}
We characterize each load profile by computing the overall load during the entire week as 
$\Gamma_i = \int_{0}^{168}l_i(t)\,dt, \forall i\in\{1,\ldots,5\}$, as well
as the average hourly load as $\bar{\ell_i} = \Gamma_i/168$, that are reported
in \dcsTabRef{load} (the upper integration limit corresponds to the number
of hours corresponding in a week).
\begin{table}
\centering
\caption{Load characteristics of each BS $i\in\mathcal{N}$.}\label{tbl:load}
\begin{tabular}{crr}
\toprule
BS & \multicolumn{1}{c}{$\Gamma_i$} & \multicolumn{1}{c}{$\bar{\ell_i}$} \\
\midrule
$1$	& $53.07$	& $0.316$ \\
$2$	& $37.17$	& $0.221$ \\
$3$	& $23.95$	& $0.143$ \\
$4$	& $40.26$	& $0.240$ \\
$5$	& $36.59$	& $0.218$ \\
\bottomrule
\end{tabular}
\end{table}

\subsection{Experimental Results}

To evaluate the performance of our algorithm, we compute 
the relative net profit increment $\mathit{RP}_i$ attained by each NO $i$,
that is defined as:
\begin{equation*}
\mathit{RP}_i = \sum_{k=1}^{A}\frac{x_i(\mathcal{S}^{(k)})}{P_i^{(k)}}-1
\end{equation*}
where $A=\lceil 168/\Delta t \rceil$ is the number of executions of the algorithm,
while $x_i(\mathcal{S}^{(k)})$ and $P_i^{(k)}$ correspond
to the payoff received by NO $i$ at the $k$-th algorithm
execution and the profit it would attain
if it worked alone (computed as in \dcsEqRef{profit}),
respectively.
Note that, in general, several Nash-stable coalitions may result
at each algorithm execution; in these cases, $x_i(\mathcal{S}^{(k)})$
is computed as the average of the payoffs yielded by all
the Nash-stable coalitions that may form.

To explain the results, we also compute, for each NO $i$, two additional
quantities, namely: 
\begin{itemize}
\item $\mathit{ON}_i$, the ratio of the number of times that BS $i$ is switched on
after an execution of the algorithm over the total number of algorithm
executions;
\item $\mathit{XL}_i$, the difference between the ratio of the total number of
users served by BS $i$ when working in a coalition over the total number of
users it would serve if it was working alone, and $1$; this quantity corresponds
to the relative deviation of load experienced by BS $i$ with respect to the case
it works alone, where a positive (negative) value represents an increment
(decrement) of load with respect to the case of working alone.
\end{itemize}

In the rest of this section, we discuss the impact of the electricity price
first (\dcsSecRef{enercost}), then we consider the effects of
the heterogeneity of user requirements (\dcsSecRef{userpop}),
then we evaluate the impact of the discretization step width
(\dcsSecRef{delta}), and finally we conclude with a discussion of our findings
(\dcsSecRef{exp-summary}).

\subsubsection{Impact of Energy Costs} \label{sec:enercost}

The energy cost $E_i$ obviously impacts on the net profit achieved by an NO $i$.
Intuitively, if energy is expensive,
NO $i$ may find it profitable to offload its users to other NOs, so that it can switch off
its BS.
Conversely, if energy is inexpensive, it may try to attract users
from other NOs.

To quantify the effects of energy price on the net profits achieved by NOs, 
we carry out experiments on a set of scenarios obtained by setting  
$E_i$ to either $E_\text{lo} = 0.12$ \$/kWh (which is a typical 
value of electricity cost in the US~\cite{EIA_EPM}) or $E_\text{hi} = 2\cdot E_\text{lo}= 0.24$ \$/kWh.
Because of space constraints, we discuss only the results corresponding
to four of the $32$ (i.e., $2^5$) distinct scenarios resulting from the
assignment of each cost value to each one of the NOs, as indicated 
in \dcsTabRef{energy-costs}.
\begin{table}
\centering
\caption{Electricity cost in the selected experimental scenarios.}\label{tbl:energy-costs}
\begin{scriptsize}
\begin{tabular}{crrrrr}
\toprule
Scenario & \multicolumn{5}{c}{Electricity Cost (\$/kWh)} \\
\cmidrule{2-6}
         & \multicolumn{1}{c}{NO $1$} & \multicolumn{1}{c}{NO $2$} & \multicolumn{1}{c}{NO $3$} & \multicolumn{1}{c}{NO $4$} & \multicolumn{1}{c}{NO $5$} \\
\midrule
$1$ & $0.12$ & $0.12$ & $0.12$ & $0.12$ & $0.12$ \\
$2$ & $0.24$ & $0.24$ & $0.24$ & $0.24$ & $0.24$ \\
$3$ & $0.12$ & $0.24$ & $0.24$ & $0.12$ & $0.12$ \\
$4$ & $0.12$ & $0.12$ & $0.12$ & $0.24$ & $0.24$ \\
\bottomrule
\end{tabular}
\end{scriptsize}
\end{table}

These four scenarios have been selected since they can be considered representative
of two opposite situations that may occur in practice: the first two scenarios 
(1 and 2) correspond indeed to standard situations where all BSs, being located in 
the same urban area, pay the same energy price, while the last two scenarios
(3 and 4) correspond to possible near-future scenarios where
different BSs can draw energy produced by either fossil fuel
or renewable sources~\cite{DeDomenico-2014-Enabling} and, as such, pay different
prices.
In these scenarios, all users require the same
minimum downlink bandwidth (i.e., $D_j=10$ Mbps) and 
generate the same revenue rate $R_j=0.07$ \$/hour (this value is
based on the $1$ GB/month \emph{Share-Everything} plan from \emph{Verizon Wireless} \cite{Verizon_ShareEverything}).
Furthermore, we assume that each NO executes an instance of \dcsAlgRef{hcf} at every hour (i.e., $\Delta t=1$ hour).
 
\dcsTabRef{p-energy-costs} shows the net profit increment $\mathit{RP_i}$ for
each NO $i$, while \dcsTabRef{on-xl-energy-costs} shows the corresponding
$\mathit{ON_i}$ and $\mathit{XL_i}$ values.
\begin{table}
\centering
\caption{Impact of energy price: net profit increments of the various NOs.}\label{tbl:p-energy-costs}
\begin{scriptsize}
\begin{tabular}{crrrrr}
\toprule
Scenario & \multicolumn{5}{c}{Net Profit Increment (\%)} \\
\cmidrule{2-6}
         & \multicolumn{1}{c}{NO $1$} & \multicolumn{1}{c}{NO $2$} & \multicolumn{1}{c}{NO $3$} & \multicolumn{1}{c}{NO $4$} & \multicolumn{1}{c}{NO $5$} \\
\midrule
$1$ & $12.10$ & $ 22.58$ & $ 40.78$  & $18.15$ & $20.67$ \\
$2$ & $39.78$ & $101.56$ & $481.44$  & $72.33$ & $90.03$ \\
$3$ & $15.31$ & $106.52$ & $504.60$  & $24.41$ & $28.37$ \\
$4$ & $15.18$ & $ 30.11$ & $ 57.08$  & $77.14$ & $96.24$ \\
\bottomrule
\end{tabular}
\end{scriptsize}
\end{table}

\begin{table}
\centering
\caption{Impact of energy price: $\mathit{ON}$ and $\mathit{XL}$ values (in \%) for the various NO.}\label{tbl:on-xl-energy-costs}
\begin{scriptsize}
\begin{tabular}{crrrrrrrrrrr}
\toprule
Scenario & \multicolumn{2}{c}{NO $1$} & \multicolumn{2}{c}{NO $2$} & \multicolumn{2}{c}{NO $3$} & \multicolumn{2}{c}{NO $4$} & \multicolumn{2}{c}{NO $5$} \\
         & $\mathit{ON}_1$ & $\mathit{XL}_1$ & $\mathit{ON}_2$ & $\mathit{XL}_2$ & $\mathit{ON}_3$ & $\mathit{XL}_3$ & $\mathit{ON}_4$ & $\mathit{XL}_4$ & $\mathit{ON}_5$ & $\mathit{XL}_5$ \\
\midrule
$1$ & $92.48$ & $60.44$ & $46.85$ & $ 33.45$ & $19.79$ & $-30.96$ & $22.03$ & $-47.55$ & $22.70$ & $-42.04$ \\
$2$ & $90.60$ & $60.43$ & $42.84$ & $ 26.00$ & $25.89$ & $- 9.80$ & $14.02$ & $-63.85$ & $23.38$ & $-32.16$ \\
$3$ & $89.88$ & $62.45$ & $ 8.63$ & $-84.90$ & $ 9.52$ & $-75.98$ & $41.67$ & $ 13.41$ & $50.00$ & $ 35.47$ \\
$4$ & $94.05$ & $76.82$ & $39.29$ & $ 23.78$ & $48.21$ & $ 56.80$ & $13.99$ & $-75.27$ & $ 5.65$ & $-87.18$ \\
\bottomrule
\end{tabular}
\end{scriptsize}
\end{table}

Let us start with scenarios $1$ and $2$.
As can be seen from the corresponding rows in \dcsTabRef{p-energy-costs},
NO $1$ and NO $3$ achieve the lowest and the
highest net profit increase, respectively, in both scenarios.
The corresponding $\mathit{ON}_i$ and $\mathit{XL}_i$ values provide an explanation 
of these facts.
Indeed, while NO $3$ is able to switch off its BS most
of the times ($\mathit{ON}_3$ is no larger than $26\%$), thus reducing its
energy cost,  NO $1$ keeps it switched on 
for $90\%$ of the times or more, thus incurring into an high energy cost.
This, in turn, is due to the characteristics of the respective load
profiles: while NO $1$ most of the times has too many clients to be able
to offload all of them to other NOs, NO $3$ is in the opposite situation.
As a consequence, NO $1$ is able to accept users coming from other NOs
($\mathit{XL}_1$ is larger than $60$\%), but the extra revenues they generate
is in large part elided by its energy costs.
Finally, the profit increments of the other NOs (namely, $2$, $4$, and $5$)
fall in between the above two extremes: they are indeed able to offload
their users more often than NO $1$ (they have a lower load), but not
as often as NO $3$ (their load is higher), as indicated by
\dcsTabRef{on-xl-energy-costs}.

This phenomenon occurs also in scenarios $3$ and $4$, where
we observe  that the largest net profit increases
are achieved by the NOs that are associated with the highest
energy costs (NOs $2$ and $3$ in scenario $3$, and NO $4$ and $5$ in
scenario $4$).
As indicated by \dcsTabRef{on-xl-energy-costs},
these NOs are indeed often able to offload all their
users to other BSs (such a frequency depends on the respective load
profiles) so that they can frequently switch 
off their BSs, and the cost savings they achieve are significant given their
higher energy costs.
Again, as in scenarios $1$ and $2$, NO $1$ gets the lower profit increase, since its
ability to switch its BS off remains lower than the other players.

\subsubsection{Impact of User Heterogeneity} \label{sec:userpop}

In the previous set of experiments, all the users were assumed
to request the same minimum downlink data rate.
However, in realistic settings, it is reasonable to expect that
users with different requirements and revenues co-exist in the same area.

In order to study the impact of the composition of the user population on the
ability of our algorithm to yield satisfactory results, we carry out
a set of experiments in which, for each one of the scenarios
listed in \dcsTabRef{energy-costs}, users are partitioned
in equal proportions
into three classes, each one characterized by a different values of the
minimum downlink data rate and revenue rate (as reported in \dcsTabRef{users}).
As in \dcsSecRef{enercost}, we assume that each NO executes an instance of \dcsAlgRef{hcf} at every hour (i.e., $\Delta t=1$ hour).
\begin{table}
\centering
\caption{Minimum downlink data rates and revenue rates for the different user classes.}\label{tbl:users}
\begin{scriptsize}
\begin{tabular}{crcr}
\toprule
User Class & \multicolumn{1}{c}{$D_j$ (Mbps)} & Traffic Type & \multicolumn{1}{c}{$R_j$ (\$/hour)} \\
\midrule
Base	   & $0.0122$                   & Voice        & $0.0175$ \\
Standard   & $0.384$                    & 3GPP         & $0.035$ \\
Premium	   & $10$                       & HSPA         & $0.07$ \\
\bottomrule
\end{tabular}
\end{scriptsize}
\end{table}

\begin{table}
\centering
\caption{Impact of user heterogeneity: net profit increments of the varios NOs.}\label{tbl:p-heter}
\begin{scriptsize}
\begin{tabular}{crrrrr}
\toprule
Scenario & \multicolumn{5}{c}{Net Profit Increment (\%)} \\
\cmidrule{2-6}
         & \multicolumn{1}{c}{NO $1$} & \multicolumn{1}{c}{NO $2$} & \multicolumn{1}{c}{NO $3$} & \multicolumn{1}{c}{NO $4$} & \multicolumn{1}{c}{NO $5$} \\
\midrule
$1$ & $ 6.23$ & $10.96$ & $19.84$ & $ 9.46$ & $10.77$ \\
$2$ & $16.90$ & $33.94$ & $78.29$ & $28.02$ & $33.82$ \\
$3$ & $ 7.95$ & $36.06$ & $81.43$ & $12.48$ & $14.84$ \\
$4$ & $ 7.88$ & $14.65$ & $28.36$ & $30.11$ & $36.05$ \\
\bottomrule
\end{tabular}
\end{scriptsize}
\end{table}

\dcsTabRef{p-heter} shows the net profit increment $\mathit{RP_i}$ for
each NO $i$.
As for the case of homogeneous users, we observe that 
in scenarios $1$ and $2$, the lowest and the
highest net profit increases are achieved by NO $1$ (i.e., the NO with the
heaviest loaded BS) and NO $3$ (i.e., the NO with the lightest loaded BS),
respectively, while, in scenarios $3$ and $4$, the largest net profit increments
are reached by the NOs that are associated with the highest energy costs (i.e.,
NOs $2$ and $3$ in scenario $3$, and NOs $4$ and $5$ in scenario $4$).
The explanation of this fact is the same given for the homogeneous users
case, so we do not repeat it here.
 
Furthermore, while the composition of user population has no appreciable
effects on the choices taken by individual NOs, it has an
evident effect on the net profit increase attained by each NO.
As can be indeed seen from \dcsTabRef{p-heter}, the $\mathit{RP}$ values
are lower than those achieved by each NO in the same scenario when
users are homogeneous (reported in \dcsTabRef{p-energy-costs}).
This is not unexpected, given the lower average revenue rate brought
by each user ($0.0356$ \$/hour versus $0.07$ \$/hour).
However, we also note that these increases remain significant.

\subsubsection{Impact of $\mathbf{\Delta t}$} \label{sec:delta}

The final set of results we comment are concerned with the
behavior of the algorithm when the discretization step $\Delta t$ increases,
i.e., when its frequency of activation decreases.
This parameter has an evident impact on the solutions computed by
the algorithm, since the larger its value, the coarser the approximation
of the load value used to compute the optimal coalition load cost $Q(\cdot)$,
and the lower the impact of the coalition formation cost $K(\cdot)$.

To quantify this impact, we run several experiments in which we progressively
increase $\Delta t$.
Because of space constraints, here we discuss only the results
obtained for $\Delta t = 2, 4, 6$ hours for the same scenarios and user population
considered in \dcsSecRef{enercost}, that are reported in \dcsFigRef{exp-deltat}.
The figure shows for each value of $\Delta t$ (in the $x$-axis) the
corresponding $\mathit{RP}$ values of the various NOs (denoted as
$\mathit{RP}(\Delta t)$, in the $y$-axis).
Also, for the sake of comparability, in the same figure, we report the
$\mathit{RP}$ values obtained for $\Delta t=1$ hour.

\begin{figure}
\centering
\subfloat[][Scenario $1$\label{fig:exp-deltat_s1}]{
\centering
\includegraphics[scale=.40]{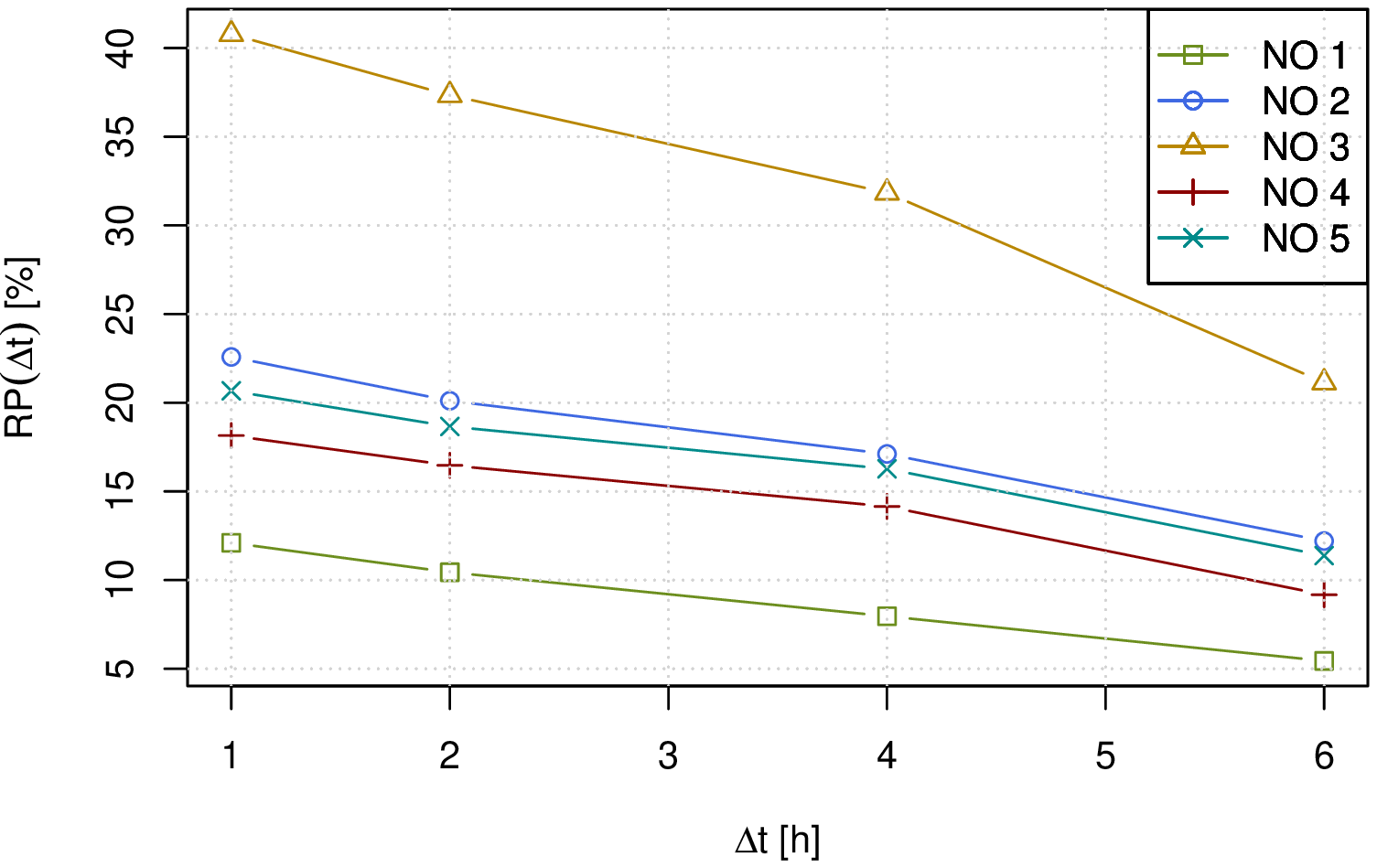}
}
~
\subfloat[][Scenario $2$\label{fig:exp-deltat_s2}]{
\centering
\includegraphics[scale=.40]{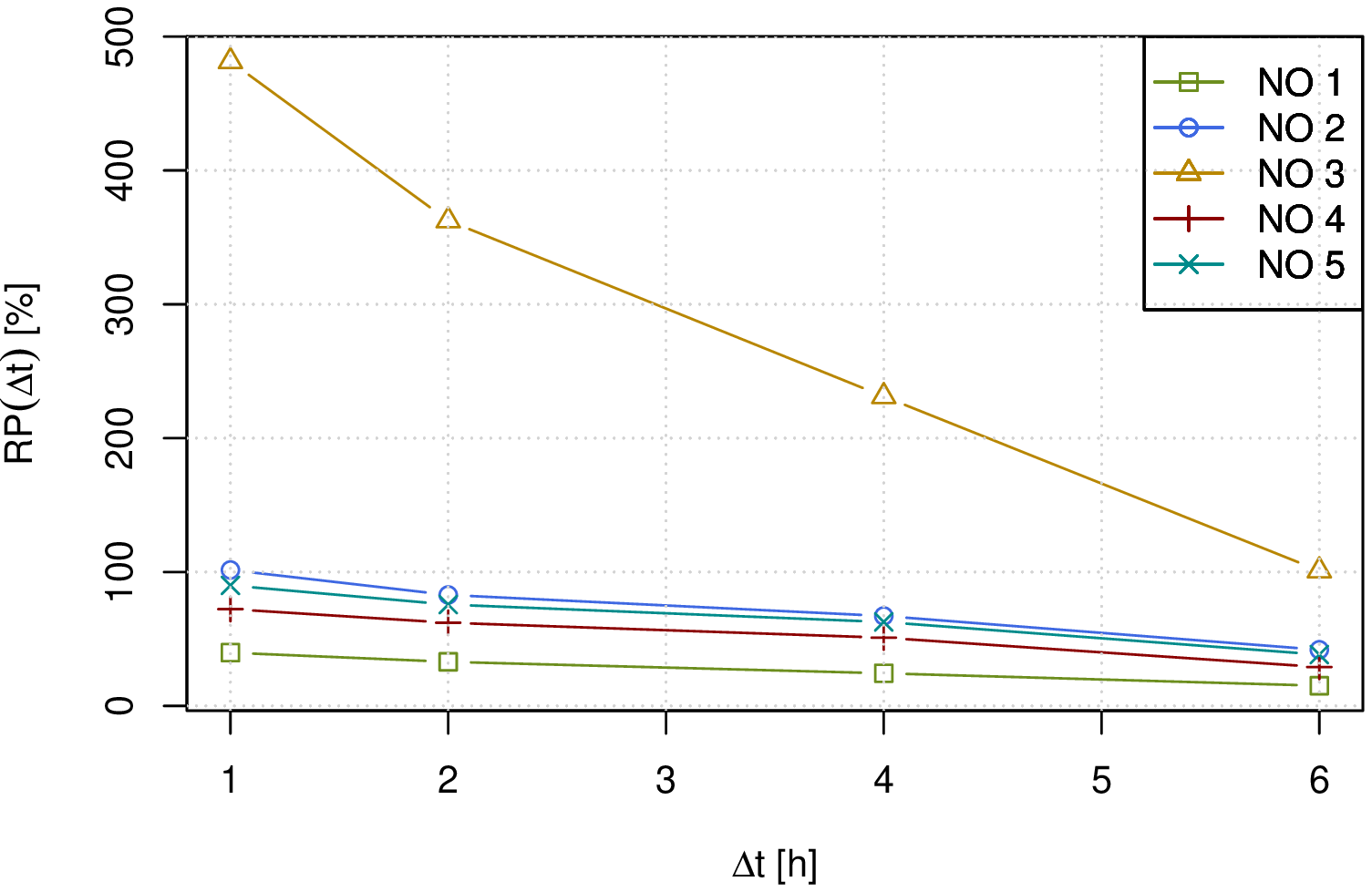}
}
\\
\subfloat[][Scenario $3$\label{fig:exp-deltat_s3}]{
\centering
\includegraphics[scale=.40]{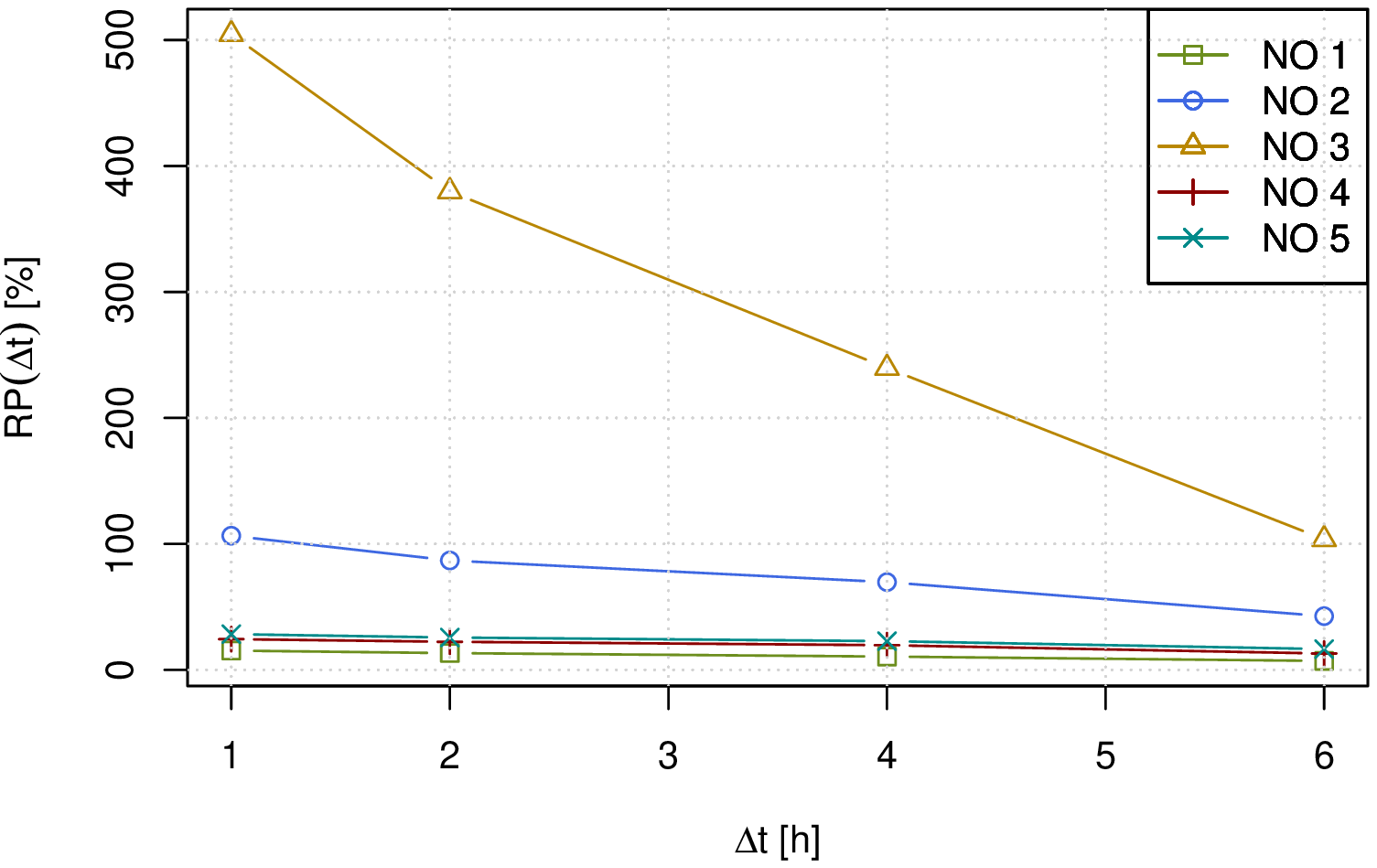}
}
~
\subfloat[][Scenario $4$\label{fig:exp-deltat_s4}]{
\centering
\includegraphics[scale=.40]{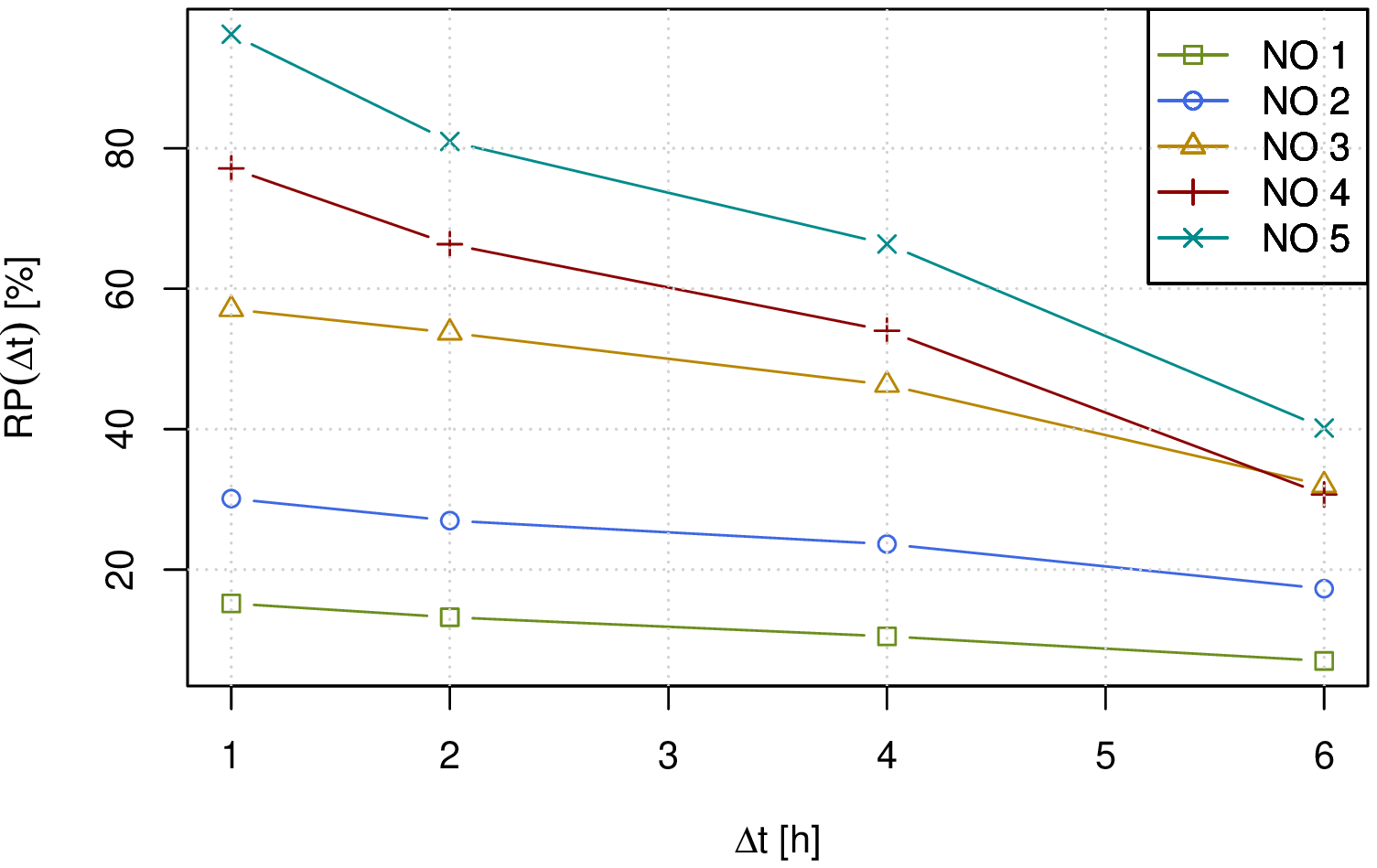}
}
\caption{Impact of $\Delta t$: $\mathit{RP}$ values for the various NOs and for $\Delta t=1,2,4,6$.}\label{fig:exp-deltat}
\end{figure}

As can be seen from \dcsFigRef{exp-deltat}, the larger $\Delta t$, the lower the net profit
increase attained by each NO for a given scenario with respect to 
when $\Delta t = 1$.

As previously pointed out for the case of $\Delta t=1$, the energy cost plays an
important role in the achieved net profit increments.
This still holds for larger value of $\Delta t$ since, as can be noted in
\dcsFigRef{exp-deltat}, the decrease of the $\mathit{RP}$ values, for
increasing values of $\Delta t$, is as much larger as the higher is the energy
cost.
For instance, this can be observed by comparing the $\mathit{RP}$ values obtained
in scenarios $1$ (see \dcsFigRef{exp-deltat_s1}) and $2$ (see
\dcsFigRef{exp-deltat_s2}), whereby the electricity cost double.
Here, we can note that as $\Delta t$ increases, the $\mathit{RP}$ values drops by
a factor of about $2$ in scenario $1$, and by a factor of nearly $2.5$ but that can
reach a factor of about $4.8$ for NO $3$, in scenario $2$.
A similar situation happens in scenarios $3$ and $4$, where the NOs that experience
the higher decrease of their $\mathit{RP}$ values are the ones with associated the
higher electricity cost, namely NOs $2$ and $3$ in scenario $3$ (see
\dcsFigRef{exp-deltat_s3}) and NOs $4$ and $5$ in scenario $4$ (see
\dcsFigRef{exp-deltat_s4}).

\subsubsection{Discussion} \label{sec:exp-summary}

From the results obtained in the experimental evaluation we can conclude that:
\begin{itemize}
\item energy costs greatly influence the coalition formation and the net profit
increments achieved by NOs;
\item NOs with higher energy costs are more motivated to join a coalition
since they can offload their users to NOs with lower energy costs, thus allowing
to switch off their BSs and hence to achieve a higher net profit increment;
\item NOs with heavier load are motivated to join a coalition as well, since
they can host users of other NOs, thus amortizing their energy costs;
\item the composition of user population has no effect on the choices made
by the various NOs, while -- if the revenue associated with each user is 
directly proportional to its minimum downlink data rate -- it has a strong
impact on the net profit increment;
\item the discretization step $\Delta t$ has a significant impact on the coalition
formation process, as the larger its value, the lower the net
profit increment;
\item when choosing the discretization step $\Delta t$, one have to also take
into account the impact of the energy cost as, the larger is $\Delta t$ the higher
is the (negative) impact of the energy cost on the achieved net profit increment.
\end{itemize}

\section{Related Works} \label{sec:related}

The problem of increasing the profit of NOs in
cellular wireless networks has been already studied in the literature,
where several papers on this topic have been published.

However, to the best of our knowledge, the problem of forming multiple stable
coalitions of BSs, in order to reduce energy consumption and to increase the
profit of different and selfish NOs, has never been tackled before.
In our work, we pursue this problem by proposing a novel approach based on
mathematical optimization and on the coalition formation game theory.

Much of current research focuses indeed on energy saving techniques as a way to
reduce NO costs \cite{DeDomenico-2014-Enabling}.
For instance, works like \cite{AjmoneMarsan-2009-Optimal,Ghazzai-2012-Optimized,Han-2013-Energy,Hasan-2013-Coalitional}
use optimization techniques and traffic profile patterns to determine when and
where to switch off BSs,
while those like \cite{Niu-2010-Cell} use cell-zooming techniques to
adaptively adjust the cell size according to traffic load and to possibly
switch off inactive cells.
Other works, like \cite{Hsu-2014-Optimizing}, focus on techniques to improve
spectrum efficiency in order to enhance the utilization of data subcarriers and
thus to better amortize the license costs of frequency bands.
These techniques, however, do not consider the cooperation among different NOs,
and therefore are unable, unlike our approach, to exploit the advantages brought
by such a cooperation. Furthermore, they do not jointly tackle the problems of
ensuring QoS to users and of reducing energy consumption. 

Approaches attempting to jointly achieve QoS and energy savings have been
recently proposed \cite{Boiardi-2013-Radio,Antoniou-2009-Access,Yaacoub-2013-Game,Zhu-2014-Pricing}.

In \cite{Boiardi-2013-Radio}, a static joint planning and management optimization
approach to limit energy consumption (by switching BSs on and off according to
the traffic load) while guaranteeing QoS and minimizing NO costs is proposed.
This approach, however, is inherently static (it operates at network design
time) so, unlike ours, is unable to operate in a dynamic environment.

In \cite{Antoniou-2009-Access}, the authors present a cooperative
game-theoretic approach, in which individual access networks with insufficient
resources join to form the grand-coalition in order to satisfy service demands.
This proposal is unable to support non trivial scenarios featuring multiple
NOs, a time-varying number of connected users, the energy consumption and the costs due to coalition
formation and operation, that may prevent the formation of
the grand coalition in favor of smaller and more stable coalitions.
Conversely, these scenarios are properly dealt with by our work.

In \cite{Yaacoub-2013-Game}, a game-theoretic approach for the
energy-efficient operation of heterogeneous LTE cellular networks, belonging to
a single NO, is proposed.
This approach, however, does not guarantee the stability of the
coalitions that are formed, while stability is a core property of
our solution. Furthermore, it is unable to deal with complex scenarios
featuring multiple NOs possibly exhibiting different energy prices, and
a time-varying population of users.
In contrast, our algorithm is able to deal with the above scenarios,
and always yield stable coalitions.

In \cite{Zhu-2014-Pricing}, the authors propose a hierarchical dynamic game
framework to increase the capacity of two-tier cellular networks by offloading
traffic from macro cells to small cells.
This work does not take into account the opportunity to selectively switch off
underutilized BSs thus offloading the related users to the remaining switched-on
BSs.
Therefore, our approach can be considered complementary to this one,
as it is able to provide a profitable way to select
what macro cells to consider before applying the proposed hierarchical game.

\section{Conclusion and Future Works} \label{sec:concl}

This paper presents a novel dynamic cooperation scheme among a group of cellular
NOs to achieve profit maximization.
We propose a cooperative game-theoretic framework to study the problem of
forming stable coalitions among NOs, and a mathematical optimization model to
allocate users to a set of BSs, in order to reduce NO costs and, at the same
time, to meet user QoS.

Our solution adopts a distributed approach in which the best solution arises
without the need to synchronize the various NOs or to resort to a trusted third
party, and such that no NO can benefit by moving from its coalition to
another (possibly empty) one.

In the proposed scheme, we model the cooperation among the NOs as a coalition
game with transferable utility and we devise a hedonic shift algorithm
to form stable coalitions.
With our algorithm, each NO autonomously and selfishly decides whether
to leave the current coalition to join a different one according to his
preference, meanwhile improving its perceived net profit.
We prove that the proposed algorithm converges to a Nash-stable partition which
determines the resulting coalition structure.
Our algorithm can be readily implemented in a distributed fashion, given that
each NO can act independently and asynchronously from any other NO in the system.
Furthermore, the asynchronicity of our algorithm makes it able to adapt to
environmental changes (like new user arrivals).

To prove the effectiveness of our approach, we perform a thorough numerical
evaluation by means of trace-based simulation, using realistic scenarios and
real-world traffic data.
To evaluate the performance of our algorithms, we vary, in a controlled way, the
values of its input parameters, like the energy cost and the user heterogeneity.

The future developments of this research is following several directions.
Firstly, we want to extend our work to include cellular networks sparsely
deployed in wider geographical areas.
To do so, we will have to take into account several issues, like the network
coverage problem, whereby a BS can be switched off only if a group of neighboring
BSs can cover the area it serves.

As as second research direction, we want to explore different variants of our
algorithm, especially suited for large cellular networks.
Specifically, when the number of BSs increases, the time to convergence of our
algorithm may be too long for practical uses, especially for small values of the
discretization step.
In these cases, it would be better to renounce to the quality of the obtained
solution in favor of a more readily available solution.
Our current algorithm always provides the best solution (i.e., a Nash-stable
partition), but at the cost of visiting, in the worst case, all possible
coalitions.
To this end, we want to design an \emph{anytime} version of our algorithm (i.e.,
an algorithm which can return a -- possibly suboptimal -- solution any time) and
we want to compare its performance with the current one.

Finally, we would like to extend our work to include the cooperation between
BSs and their users, in order to improve the energy reduction of BSs and the
quality of experience of connected users.

\bibliographystyle{model1-num-names}

\begin{thebibliography}{40}
\expandafter\ifx\csname natexlab\endcsname\relax\def\natexlab#1{#1}\fi
\providecommand{\url}[1]{\texttt{#1}}
\providecommand{\href}[2]{#2}
\providecommand{\path}[1]{#1}
\providecommand{\DOIprefix}{doi:}
\providecommand{\ArXivprefix}{arXiv:}
\providecommand{\URLprefix}{URL: }
\providecommand{\Pubmedprefix}{pmid:}
\providecommand{\doi}[1]{\href{http://dx.doi.org/#1}{\path{#1}}}
\providecommand{\Pubmed}[1]{\href{pmid:#1}{\path{#1}}}
\providecommand{\bibinfo}[2]{#2}
\ifx\xfnm\relax \def\xfnm[#1]{\unskip,\space#1}\fi
\bibitem[{{Global Action Plan}(2007)}]{GAP-07}
\bibinfo{author}{{Global Action Plan}}, \bibinfo{title}{An inefficient truth},
  \bibinfo{howpublished}{Online: \url{http://www.globalactionplan.org.uk}},
  \bibinfo{year}{2007}.
\bibitem[{{The Climate Group} and {GeSI}(2008)}]{ClimateGroup-2008-Smart2020}
\bibinfo{author}{{The Climate Group}}, \bibinfo{author}{{GeSI}},
  \bibinfo{title}{{SMART} 2020: Enabling the low carbon economy in the
  information age}, \bibinfo{howpublished}{Online:
  \url{http://www.theclimategroup.org/\_assets/files/Smart2020Report.pdf}},
  \bibinfo{year}{2008}.
\bibitem[{Vereecken et~al.(2011)Vereecken, Heddeghem, Deruyck, Puype, Lannoo,
  Joseph, Colle, Martens, and Demeester}]{Vereecken-2011-Power}
\bibinfo{author}{W.~Vereecken}, \bibinfo{author}{W.~V. Heddeghem},
  \bibinfo{author}{M.~Deruyck}, \bibinfo{author}{B.~Puype},
  \bibinfo{author}{B.~Lannoo}, \bibinfo{author}{W.~Joseph},
  \bibinfo{author}{D.~Colle}, \bibinfo{author}{L.~Martens},
  \bibinfo{author}{P.~Demeester},
\newblock \bibinfo{title}{Power consumption in telecommunication networks:
  overview and reduction strategies},
\newblock \bibinfo{journal}{IEEE Comm Mag} \bibinfo{volume}{49}
  (\bibinfo{year}{2011}) \bibinfo{pages}{62--69}.
\bibitem[{Oh et~al.(2011)Oh, Krishnamachari, Liu, and Niu}]{Oh-2011-Toward}
\bibinfo{author}{E.~Oh}, \bibinfo{author}{B.~Krishnamachari},
  \bibinfo{author}{X.~Liu}, \bibinfo{author}{Z.~Niu},
\newblock \bibinfo{title}{Toward dynamic energy-efficient operation of cellular
  network infrastructure},
\newblock \bibinfo{journal}{IEEE Comm Mag} \bibinfo{volume}{49}
  (\bibinfo{year}{2011}) \bibinfo{pages}{56--61}.
\bibitem[{Domenico et~al.(2014)Domenico, Strinati, and
  Capone}]{DeDomenico-2014-Enabling}
\bibinfo{author}{A.~D. Domenico}, \bibinfo{author}{E.~C. Strinati},
  \bibinfo{author}{A.~Capone},
\newblock \bibinfo{title}{Enabling green cellular networks: A survey and
  outlook},
\newblock \bibinfo{journal}{Comput Comm} \bibinfo{volume}{37}
  (\bibinfo{year}{2014}) \bibinfo{pages}{5--24}.
\bibitem[{Auer et~al.(2011)Auer, Giannini, Desset, Godor, Skillermark, Olsson,
  Imran, Sabella, Gonzalez, Blume, and Fehske}]{Auer2011}
\bibinfo{author}{G.~Auer}, \bibinfo{author}{V.~Giannini},
  \bibinfo{author}{C.~Desset}, \bibinfo{author}{I.~Godor},
  \bibinfo{author}{P.~Skillermark}, \bibinfo{author}{M.~Olsson},
  \bibinfo{author}{M.~Imran}, \bibinfo{author}{D.~Sabella},
  \bibinfo{author}{M.~Gonzalez}, \bibinfo{author}{O.~Blume},
  \bibinfo{author}{A.~Fehske},
\newblock \bibinfo{title}{How much energy is needed to run a wireless
  network?},
\newblock \bibinfo{journal}{IEEE Wireless Comm Mag} \bibinfo{volume}{18}
  (\bibinfo{year}{2011}) \bibinfo{pages}{40--49}.
\bibitem[{Peng et~al.(2011)Peng, Lee, Lu, Luo, and Li}]{Peng-2011-Traffic}
\bibinfo{author}{C.~Peng}, \bibinfo{author}{S.-B. Lee},
  \bibinfo{author}{S.~Lu}, \bibinfo{author}{H.~Luo}, \bibinfo{author}{H.~Li},
\newblock \bibinfo{title}{Traffic-driven power saving in operational 3g
  cellular networks},
\newblock in: \bibinfo{booktitle}{Proc. of the $17^\text{th}$ Annual Int.
  Conference on Mobile Computing and Networking (MobiCom)},
  \bibinfo{publisher}{ACM}, \bibinfo{year}{2011}, pp.
  \bibinfo{pages}{121--132}.
\bibitem[{Micallef et~al.(2010)Micallef, Mogensen, and
  Scheck}]{Micallef-2010-Cell}
\bibinfo{author}{G.~Micallef}, \bibinfo{author}{P.~Mogensen},
  \bibinfo{author}{H.-O. Scheck},
\newblock \bibinfo{title}{Cell size breathing and possibilities to introduce
  cell sleep mode},
\newblock in: \bibinfo{booktitle}{Proc. of the European Wireless Conference
  (EW)}, \bibinfo{year}{2010}, pp. \bibinfo{pages}{111--115}.
\bibitem[{Niu et~al.(2010)Niu, Wu, Gong, and Yang}]{Niu-2010-Cell}
\bibinfo{author}{Z.~Niu}, \bibinfo{author}{Y.~Wu}, \bibinfo{author}{J.~Gong},
  \bibinfo{author}{Z.~Yang},
\newblock \bibinfo{title}{Cell zooming for cost-efficient green cellular
  networks},
\newblock \bibinfo{journal}{IEEE Comm Mag}  (\bibinfo{year}{2010}).
\bibitem[{Peleg and Sudh{\"{o}}lter(2007)}]{Peleg-2007-CooperativeGames}
\bibinfo{author}{B.~Peleg}, \bibinfo{author}{P.~Sudh{\"{o}}lter},
  \bibinfo{title}{Introduction to the Theory of Cooperative Games},
  \bibinfo{edition}{$2^\text{nd}$} ed., \bibinfo{publisher}{Springer Berlin
  Heidelberg}, \bibinfo{year}{2007}.
\bibitem[{Dr{\`e}ze and Greenberg(1980)}]{DrezeGreenberg1980}
\bibinfo{author}{J.~Dr{\`e}ze}, \bibinfo{author}{J.~Greenberg},
\newblock \bibinfo{title}{Hedonic coalitions: Optimality and stability},
\newblock \bibinfo{journal}{Econometrica} \bibinfo{volume}{48}
  (\bibinfo{year}{1980}) \bibinfo{pages}{987--1003}.
\bibitem[{Bogomolnaia and Jackson(2002)}]{BOGOMONLAIA-JACKSON_2002}
\bibinfo{author}{A.~Bogomolnaia}, \bibinfo{author}{M.~Jackson},
\newblock \bibinfo{title}{{The Stability of Hedonic Coalition Structures}},
\newblock \bibinfo{journal}{Game Econ Behav} \bibinfo{volume}{38}
  (\bibinfo{year}{2002}) \bibinfo{pages}{201--230}.
\bibitem[{{Ajmone Marsan} et~al.(2012){Ajmone Marsan}, Chiaraviglio, Ciullo,
  and Meo}]{AjmoneMarsan-2012-Multiple}
\bibinfo{author}{M.~{Ajmone Marsan}}, \bibinfo{author}{L.~Chiaraviglio},
  \bibinfo{author}{D.~Ciullo}, \bibinfo{author}{M.~Meo},
\newblock \bibinfo{title}{Multiple daily base station switch-offs in cellular
  networks},
\newblock in: \bibinfo{booktitle}{Proc. of the $4^\text{th}$ Int. Conference on
  Communications and Electronics (ICCE)}, \bibinfo{year}{2012}, pp.
  \bibinfo{pages}{245--250}.
\bibitem[{Pollakis et~al.(2012)Pollakis, Cavalcante, and
  Sta{\'n}czak}]{Pollakis-2012-Base}
\bibinfo{author}{E.~Pollakis}, \bibinfo{author}{R.~L. Cavalcante},
  \bibinfo{author}{S.~Sta{\'n}czak},
\newblock \bibinfo{title}{Base station selection for energy efficient network
  operation with the majorization-minimization algorithm},
\newblock in: \bibinfo{booktitle}{Proc. of the $13^\text{th}$ IEEE Int.
  Workshop on Signal Processing Advances in Wireless Communications (SPAWC)},
  \bibinfo{year}{2012}.
\bibitem[{Deruyck et~al.(2012)Deruyck, Tanghe, Joseph, and
  Martens}]{Deruyck-2012-Characterization}
\bibinfo{author}{M.~Deruyck}, \bibinfo{author}{E.~Tanghe},
  \bibinfo{author}{W.~Joseph}, \bibinfo{author}{L.~Martens},
\newblock \bibinfo{title}{Characterization and optimization of the power
  consumption in wireless access networks by taking daily traffic variations
  into account},
\newblock \bibinfo{journal}{EURASIP J Wireless Comm Networking}
  \bibinfo{volume}{2012} (\bibinfo{year}{2012}) \bibinfo{pages}{1--12}.
\bibitem[{Lorincz et~al.(2012)Lorincz, Garma, and
  Petrovic}]{Lorincz-2012-Measurements}
\bibinfo{author}{J.~Lorincz}, \bibinfo{author}{T.~Garma},
  \bibinfo{author}{G.~Petrovic},
\newblock \bibinfo{title}{{Measurements and Modelling of Base Stations Power
  Consumption under Real Traffic Loads}},
\newblock \bibinfo{journal}{Sensors} \bibinfo{volume}{12}
  (\bibinfo{year}{2012}) \bibinfo{pages}{4281--4310}.
\bibitem[{Conte et~al.(2011)Conte, Feki, Chiaraviglio, Ciullo, Meo, and {Ajmone
  Marsan}}]{Conte-2011-Cell}
\bibinfo{author}{A.~Conte}, \bibinfo{author}{A.~Feki},
  \bibinfo{author}{L.~Chiaraviglio}, \bibinfo{author}{D.~Ciullo},
  \bibinfo{author}{M.~Meo}, \bibinfo{author}{M.~{Ajmone Marsan}},
\newblock \bibinfo{title}{Cell wilting and blossoming for energy efficiency},
\newblock \bibinfo{journal}{IEEE Wireless Comm Mag} \bibinfo{volume}{18}
  (\bibinfo{year}{2011}) \bibinfo{pages}{50--57}.
\bibitem[{Willkomm et~al.(2008)Willkomm, Machiraju, Bolot, and
  Wolisz}]{Willkomm-2008-Primary}
\bibinfo{author}{D.~Willkomm}, \bibinfo{author}{S.~Machiraju},
  \bibinfo{author}{J.~Bolot}, \bibinfo{author}{A.~Wolisz},
\newblock \bibinfo{title}{Primary users in cellular networks: A large-scale
  measurement study},
\newblock in: \bibinfo{booktitle}{Proc. of the $3^\text{rd}$ IEEE Symposium on
  New Frontiers in Dynamic Spectrum Access Networks (DySPAN)},
  \bibinfo{year}{2008}, pp. \bibinfo{pages}{1--11}.
\bibitem[{Ray(2007)}]{Book_RAY2007}
\bibinfo{author}{D.~Ray}, \bibinfo{title}{{A Game-Theoretic Perspective on
  Coalition Formation}}, The Lipsey Lectures, \bibinfo{publisher}{Oxford
  University Press}, \bibinfo{year}{2007}.
\bibitem[{Shapley(1953)}]{SHAPLEY_53}
\bibinfo{author}{L.~S. Shapley},
\newblock \bibinfo{title}{{A Value for $n$-person Games}},
\newblock in: \bibinfo{editor}{H.~Kuhn}, \bibinfo{editor}{A.~Tucker} (Eds.),
  \bibinfo{booktitle}{{Contributions to the Theory of Games}},
  \bibinfo{publisher}{Princeton University Press}, \bibinfo{year}{1953}, pp.
  \bibinfo{pages}{307--317}.
\bibitem[{Aumann and Dr{\'e}ze(1974)}]{Aumann-1974-Cooperative}
\bibinfo{author}{R.~Aumann}, \bibinfo{author}{J.~Dr{\'e}ze},
\newblock \bibinfo{title}{Cooperative games with coalition structures},
\newblock \bibinfo{journal}{Int J Game Theor} \bibinfo{volume}{3}
  (\bibinfo{year}{1974}) \bibinfo{pages}{217--237}.
\bibitem[{Saad et~al.(2011)Saad, Han, Ba{\c{s}}ar, Debbah, and
  Hj{\o}rungnes}]{Saad-2011-Hedonic}
\bibinfo{author}{W.~Saad}, \bibinfo{author}{Z.~Han},
  \bibinfo{author}{T.~Ba{\c{s}}ar}, \bibinfo{author}{M.~Debbah},
  \bibinfo{author}{A.~Hj{\o}rungnes},
\newblock \bibinfo{title}{Hedonic coalition formation for distributed task
  allocation among wireless agents},
\newblock \bibinfo{journal}{IEEE Trans Mobile Comput} \bibinfo{volume}{10}
  (\bibinfo{year}{2011}) \bibinfo{pages}{1327--1344}.
\bibitem[{Kshemkalyani and Singhal(2008)}]{Kshemkalyani-2008-Distributed}
\bibinfo{author}{A.~Kshemkalyani}, \bibinfo{author}{M.~Singhal},
  \bibinfo{title}{{Distributed Computing: Principles, Algorithms, and
  Systems}}, \bibinfo{publisher}{Cambridge University Press},
  \bibinfo{year}{2008}.
\bibitem[{Weiss(2013)}]{Weiss-2013-MAS}
\bibinfo{editor}{G.~Weiss} (Ed.), \bibinfo{title}{Multiagent Systems},
  \bibinfo{edition}{$2^\text{nd}$} ed., \bibinfo{publisher}{MIT Press},
  \bibinfo{year}{2013}.
\bibitem[{Ballester(2004)}]{Ballester-2004-NP}
\bibinfo{author}{C.~Ballester},
\newblock \bibinfo{title}{{NP}-completeness in hedonic games},
\newblock \bibinfo{journal}{Game Econ Behav} \bibinfo{volume}{49}
  (\bibinfo{year}{2004}) \bibinfo{pages}{1--30}.
\bibitem[{Lorincz et~al.(2012)Lorincz, Capone, and
  Begusic}]{Lorincz-2012-Impact}
\bibinfo{author}{J.~Lorincz}, \bibinfo{author}{A.~Capone},
  \bibinfo{author}{D.~Begusic},
\newblock \bibinfo{title}{Impact of service rates and base station switching
  granularity on energy consumption of cellular networks},
\newblock \bibinfo{journal}{EURASIP J Wireless Comm Networking}
  \bibinfo{volume}{342} (\bibinfo{year}{2012}) \bibinfo{pages}{1--24}.
\bibitem[{Boiardi et~al.(2013)Boiardi, Capone, and
  Sans{\'o}}]{Boiardi-2013-Radio}
\bibinfo{author}{S.~Boiardi}, \bibinfo{author}{A.~Capone},
  \bibinfo{author}{B.~Sans{\'o}},
\newblock \bibinfo{title}{Radio planning of energy-aware cellular networks},
\newblock \bibinfo{journal}{Comput Network} \bibinfo{volume}{57}
  (\bibinfo{year}{2013}) \bibinfo{pages}{2564--2577}.
\bibitem[{CPL(2014)}]{CPLEX}
\bibinfo{title}{{IBM ILOG CPLEX Optimizer}}, \bibinfo{howpublished}{Online:
  \url{http://www.ibm.com/software/integration/optimization/cplex-optimizer}},
  \bibinfo{year}{2014}.
\bibitem[{Son et~al.(2011{\natexlab{a}})Son, Oh, and
  Krishnamachari}]{ANRG_Data}
\bibinfo{author}{K.~Son}, \bibinfo{author}{E.~Oh},
  \bibinfo{author}{B.~Krishnamachari}, \bibinfo{title}{Normalized cellular
  traffic during one week}, \bibinfo{howpublished}{Online:
  \url{http://anrg.usc.edu/downloads}}, \bibinfo{year}{2011}{\natexlab{a}}.
  \bibinfo{note}{Dataset No. 18}.
\bibitem[{Son et~al.(2011{\natexlab{b}})Son, Kim, Yi, and
  Krishnamachari}]{Son-2011-Base}
\bibinfo{author}{K.~Son}, \bibinfo{author}{H.~Kim}, \bibinfo{author}{Y.~Yi},
  \bibinfo{author}{B.~Krishnamachari},
\newblock \bibinfo{title}{Base station operation and user association
  mechanisms for energy-delay tradeoffs in green cellular networks},
\newblock \bibinfo{journal}{IEEE J Sel Area Comm} \bibinfo{volume}{29}
  (\bibinfo{year}{2011}{\natexlab{b}}) \bibinfo{pages}{1525--1536}.
\bibitem[{{U.S. Energy Information Administration}(2013)}]{EIA_EPM}
\bibinfo{author}{{U.S. Energy Information Administration}},
  \bibinfo{title}{Average retail price of electricity to ultimate customers by
  end-use sector}, \bibinfo{howpublished}{Online:
  \url{http://www.eia.gov/electricity}}, \bibinfo{year}{2013}.
\bibitem[{{Verizon Wireless}(2014)}]{Verizon_ShareEverything}
\bibinfo{author}{{Verizon Wireless}}, \bibinfo{title}{{Share Everthing} plan},
  \bibinfo{howpublished}{Online:
  \url{http://www.verizonwireless.com/wcms/consumer/shop/share-everything.html}},
  \bibinfo{year}{2014}.
\bibitem[{{Ajmone Marsan} et~al.(2009){Ajmone Marsan}, Chiaraviglio, Ciullo,
  and Meo}]{AjmoneMarsan-2009-Optimal}
\bibinfo{author}{M.~{Ajmone Marsan}}, \bibinfo{author}{L.~Chiaraviglio},
  \bibinfo{author}{D.~Ciullo}, \bibinfo{author}{M.~Meo},
\newblock \bibinfo{title}{Optimal energy savings in cellular access networks},
\newblock in: \bibinfo{booktitle}{Proc. of the IEEE Int. Conference on
  Communications Workshops (ICC)}, \bibinfo{year}{2009}, pp.
  \bibinfo{pages}{1--5}.
\bibitem[{Ghazzai et~al.(2012)Ghazzai, Yaacoub, Alouini, and
  Abu-Dayya}]{Ghazzai-2012-Optimized}
\bibinfo{author}{H.~Ghazzai}, \bibinfo{author}{E.~Yaacoub},
  \bibinfo{author}{M.~Alouini}, \bibinfo{author}{A.~Abu-Dayya},
\newblock \bibinfo{title}{Optimized green operation of {LTE} networks in the
  presence of multiple electricity providers},
\newblock in: \bibinfo{booktitle}{Proc. of the IEEE Globecom Workshops (GC
  Wkshps)}, \bibinfo{year}{2012}, pp. \bibinfo{pages}{664--669}.
\bibitem[{Han et~al.(2013)Han, Z., and {Ray Liu}}]{Han-2013-Energy}
\bibinfo{author}{F.~Han}, \bibinfo{author}{Z.~S. Z.}, \bibinfo{author}{K.~{Ray
  Liu}},
\newblock \bibinfo{title}{Energy-efficient base-station cooperative operation
  with guaranteed {QoS}},
\newblock \bibinfo{journal}{IEEE Trans Comm} \bibinfo{volume}{61}
  (\bibinfo{year}{2013}) \bibinfo{pages}{3505--3517}.
\bibitem[{Hasan et~al.(2013)Hasan, Altman, and Gorce}]{Hasan-2013-Coalitional}
\bibinfo{author}{C.~Hasan}, \bibinfo{author}{E.~Altman}, \bibinfo{author}{J.-M.
  Gorce},
\newblock \bibinfo{title}{The coalitional switch-off game of service
  providers},
\newblock in: \bibinfo{booktitle}{Proc. of the $9^\text{th}$ IEEE Int.
  Conference on Wireless and Mobile Computing, Networking and Communications
  (WiMob)}, \bibinfo{year}{2013}.
\bibitem[{Hsu et~al.(2013)Hsu, Chang, Chou, and Abichar}]{Hsu-2014-Optimizing}
\bibinfo{author}{C.-C. Hsu}, \bibinfo{author}{J.~M. Chang},
  \bibinfo{author}{Z.-T. Chou}, \bibinfo{author}{Z.~Abichar},
\newblock \bibinfo{title}{Optimizing spectrum-energy efficiency in downlink
  cellular networks},
\newblock \bibinfo{journal}{IEEE Trans Mobile Comput}  (\bibinfo{year}{2013}).
  \bibinfo{note}{In press}.
\bibitem[{Antoniou et~al.(2009)Antoniou, Koukoutsidis, Jaho, Pitsillides, and
  Stavrakakis}]{Antoniou-2009-Access}
\bibinfo{author}{J.~Antoniou}, \bibinfo{author}{I.~Koukoutsidis},
  \bibinfo{author}{E.~Jaho}, \bibinfo{author}{A.~Pitsillides},
  \bibinfo{author}{I.~Stavrakakis},
\newblock \bibinfo{title}{Access network synthesis game in next generation
  networks},
\newblock \bibinfo{journal}{Comput Network} \bibinfo{volume}{53}
  (\bibinfo{year}{2009}) \bibinfo{pages}{2716--2726}.
\bibitem[{Yaacoub et~al.(2013)Yaacoub, Imran, Dawy, and
  Abu-Dayya}]{Yaacoub-2013-Game}
\bibinfo{author}{E.~Yaacoub}, \bibinfo{author}{A.~Imran},
  \bibinfo{author}{Z.~Dawy}, \bibinfo{author}{A.~Abu-Dayya},
\newblock \bibinfo{title}{A game theoretic framework for energy efficient
  deployment and operation of heterogeneous {LTE} networks},
\newblock in: \bibinfo{booktitle}{Proc. of the {$18^\text{th}$} IEEE Int.
  Workshop on Computer Aided Modeling and Design of Communication Links and
  Networks (CAMAD)}, \bibinfo{year}{2013}, pp. \bibinfo{pages}{33--37}.
\bibitem[{Zhu et~al.(2013)Zhu, Hossain, and Niyato}]{Zhu-2014-Pricing}
\bibinfo{author}{K.~Zhu}, \bibinfo{author}{E.~Hossain},
  \bibinfo{author}{D.~Niyato},
\newblock \bibinfo{title}{Pricing, spectrum sharing, and service selection in
  two-tier small cell networks: A hierarchical dynamic game approach},
\newblock \bibinfo{journal}{IEEE Trans Mobile Comput}  (\bibinfo{year}{2013}).
  \bibinfo{note}{In press}.

\end{thebibliography}

\end{document}